\documentclass[leqno, 11pt]{amsart}
\usepackage{color}
\usepackage{url}

\usepackage{bbm}
\newcommand{\eqmean}[2]{\langle #1 \rangle^{(\text{eq},#2)}}

\newcommand{\mean}[1]{\langle #1 \rangle}
\newcommand{\ci}{{\rm i}}
\newcommand{\rmd}{{\rm d}}
\newcommand{\rme}{{\rm e}}
\newcommand{\Z}{{\mathbb Z}}
\newcommand{\R}{{\mathbb R}}

\newcommand{\vep}{\varepsilon}

\newcommand{\cf}{\mathbbm{1}}

\newcommand{\re}{{\rm Re\,}}
\newcommand{\im}{{\rm Im\,}}

\newcommand{\muss}{\mu_{\text{s,s}}}
\newcommand{\musv}{\mu_{\text{s,v}}}
\newcommand{\nessvf}[1]{\langle #1 \rangle_{\text{s,v}}}
\newcommand{\nesssc}[1]{\langle #1 \rangle_{\text{s,s}}}

\usepackage{pdfsync}
\usepackage{enumerate}
\usepackage{graphicx}
\usepackage{multirow}
\numberwithin{equation}{section}

\newcommand\bq{{\mathbf q}}
\newcommand\bp{{\mathbf p}}

\newcommand\br{{\mathbf r}}

\newcommand\RR{{\mathbb R}}

\newcommand\ZZ{{\mathbb Z}}

\newcommand\bT{{\mathbf T}}

\newcommand\ve{\varepsilon}

\newcommand{\mc}[1]{{\mathcal #1}}

\newcommand{\bb}[1]{{\mathbb #1}}

\usepackage{amssymb, amsmath}

\newtheorem{prop}{Proposition}

\newtheorem{lemma}{Lemma}

\newtheorem{cor}{Corollary}

\title[Harmonic Chains with Bulk Noises]{Harmonic Systems with Bulk Noises}

\date{\today}


\begin{document}

\begin{abstract}
We consider a harmonic chain in contact with thermal reservoirs at different temperatures and subject to bulk noises of different types: velocity flips or self-consistent reservoirs.  While both systems have the same covariances in the nonequilibrium stationary state (NESS) the measures are very different. We study hydrodynamical scaling, large deviations, fluctuations, and long range correlations in both systems. Some of our results extend to higher dimensions. 
\end{abstract}

\author{C.~Bernardin}
\address{C. Bernardin\\
Universit\'e de Lyon and CNRS,UMPA, UMR-CNRS 5669, ENS-Lyon,\\
46, all\'ee d'Italie, \\
69364 Lyon Cedex 07 \\
France
}
\email{Cedric.Bernardin@umpa.ens-lyon.fr}

\author{V.~Kannan}
\address{Venkateshan Kannan\\ Department of Mathematics and Physics\\ Rutgers University\\
Piscataway, NJ\\ USA}
\email{kven@physics.rutgers.edu}

\author{J.~L.~Lebowitz}
\address{Joel L. Lebowitz\\ Department of Mathematics and Physics\\ Rutgers University\\
Piscataway, NJ\\ USA}
\email{lebowitz@math.rutgers.edu}

  \author{J.~Lukkarinen}
\address{Jani Lukkarinen\\
Department of Mathematics and Statistics\\
University of Helsinki\\
P.O.~Box 68\\
00014~Helsingin yliopisto\\
Finland
}
\email{jani.lukkarinen@helsinki.fi}

\maketitle

\section{Introduction}
\label{sec:models}
We consider a harmonic chain consisting of $N$ oscillators. The positions and momenta of the oscillators are denoted by $q_j$ and $p_j$ respectively. The Hamiltonian is given by
\begin{equation}
\label{eq:energy}
{\mc H} =\sum_{j=0}^{N+1} \cfrac{p_j^2}{2}\;  + \;  \nu^2 \sum_{j=0}^{N+1} \cfrac{ q_j^2}{2} \; +\sum_{\substack{|i-j|=1,\\ i,j \in \{0, \ldots, N+1\}}}\cfrac{(q_{j}-q_i)^2}{4}
\end{equation}
where we have set the mass of each particle and the nearest neighbor coupling equal to $1$. We impose boundary conditions $q_0=q_{N+1}=0$, $p_0=p_{N+1}=0$. The strength of the pinning potential is regulated by the parameter $\nu \ge 0$.
The energy of site $j$ is given  by
\begin{equation}
\label{eq:energyj}
{\mc E}_j = \cfrac{p_j^2}{2} + \nu^2 \cfrac{q_j^2}{2} + \cfrac{1}{4} \sum_{i:\, |i-j|=1} (q_i -q_j)^2.
\end{equation}

In the absence of any stochastic noise the generator of the dynamics is given by the Liouville operator ${\mc A}$
 \begin{equation}
 \label{eq:liouvilleA}
{\mc A} =\sum_{j=1}^N \left\{ p_j \partial_{q_j} - \left[ (\nu^2 -\Delta) q\right]_j \partial_{p_j} \right\}
\end{equation}
where $\Delta$ is the discrete Laplacian with Dirichlet boundary conditions.


 We shall now consider two ways of adding noise to the system. The first case we consider is the so-called self-consistent model \cite{BLL04,BLLO}. Each site $j \in \{1,\ldots,N\}$ is connected to a Langevin reservoir at temperature $T_j >0$. The generator of the Langevin reservoir at temperature $T_j$ acting on the particle at $j$ is given by
\begin{equation}
\label{eq:thermostat}
{\mc B}_{j,T_j} = T_j \partial_{p_j}^2 - p_j \partial_{p_j}.
\end{equation}
The temperatures of the reservoirs at the boundary sites are fixed by the conditions  $T_1=T_\ell$, $T_N =T_r$, while the temperatures of the interior site reservoirs are determined self-consistently, by requiring that in the non-equilibrium stationary state (NESS) there is no net flux of energy between the system and the interior reservoirs.

The second model we consider is one in which only particle $1$ and particle $N$ are in contact with heat baths. In the bulk we add a flip dynamics which consists of reversing the velocity of each particle at random independent Poissonian times.  These flips are energy conserving and we call the model the velocity flip model.

In both situations the generator of the total dynamics can be written as
\begin{equation}
\label{eq:generatorL}
{\mc L} = {\mc A} +\gamma {\mc S} + {\mc B}_{1,T_\ell} + {\mc B}_{N,T_r}
\end{equation}
where $\gamma>0$ is the intensity of the bulk noise generator ${\mc S}$. The phase space is $\Omega_N = \RR^N \times \RR^N$ and a configuration $(q_1, \ldots, q_N, p_1, \ldots, p_N)$ is denoted by $\omega=(\bq, \bp)$. The configuration at time $t$ of the process is denoted by $\omega (t)$.

For the self-consistent chain ${\mc S}$ has the following expression
\begin{equation*}
{\mc S}= \sum_{j=2}^{N-1} {\mc B}_{j, T_j}\, ,
\end{equation*}
It can be shown that for any given temperatures ${\bT} := \{ T_j>0\, ; \, j=1, \ldots,N\}$ there exists a unique steady state $\muss ^{N,{\bT}}$ (which depends on the $T_j$'s) for the dynamics generated by ${\mc L}$. Moreover, it can be proved that there exists a unique ${\bT}:={\bT}^{\rm{sc}}$ such that the self-consistency condition $\muss ^{N,\bT} (p_j^2) = T_j$ for $j=2, \ldots,N-1$ is satisfied.  (If $\nu>0$, the proof is given in \cite{BLL04}; the case $\nu=0$ is treated later, in Section \ref{sec:scmunpinned}.)  In the sequel we fix the temperatures to follow the self-consistent profile ${\bT}$ and denote the corresponding stationary state by $\muss^N$ or $\nesssc{\cdot}$.

For the velocity flip model the generator ${\mc S}$ is given by
\begin{equation*}
({\mc S}f)(\omega)=\cfrac{1}{2}\sum_{j=1}^N \left[ f(\omega^j) -f(\omega)\right]
\end{equation*}
where $\omega^j$ is the configuration obtained from $\omega=(\bq,\bp)$ by flipping $p_j$ to $-p_j$.
Following the lines of the proof given in \cite{BO} or \cite{BO2} one can prove that the velocity flip model has a unique non-equilibrium stationary state denoted by $\musv^N$ or $\nessvf{\cdot}$.


The outline of the rest of the paper is as follows. In Section \ref{sec:ss} we investigate the similarities between the NESS of the self-consistent model and the NESS of the velocity-flip model and we show that while the two probability measures have the same covariances, $\nesssc{\cdot}$ is Gaussian but $\nessvf{\cdot}$ is a non-trivial mixture of Gaussian states. In Section \ref{sec:hl} we obtain the hydrodynamical equations for the velocity flip model. In Section \ref{sec:cor} we investigate the large deviation fuction and the stationary fluctuations of the energy. In particular, we show that long-range correlations are present in the state $\nessvf{\cdot}$. This is very different from the NESS of the self-consistent chain: we show in Section \ref{sec:corsc} that long-range correlations are not present in the state $\nesssc{\cdot}$. 
Some technical matters are deferred to the Appendices. A summary of our results
is given in \cite{BKLL11a}.

\section{The steady states}
\label{sec:ss}

If the temperatures $T_\ell$ and $T_r$ are equal to a common value $T$, the steady state of the self-consistent chain and the steady state of the velocity flip model are both equal to the Gibbs state with temperature $T$ that we denote by $\mu_T^N$ or $\eqmean{\cdot}{T}$. This is a Gaussian measure with covariance $C_{{\rm{eq}}} (T) = T C_{{\rm {eq}}} (1)$.

In contrast, for $T_\ell \ne T_r$, the NESS of the two models are different.  However,
it is easy to see that both stationary states are centered, and by the results derived in \cite{DKL}, we know that the two point correlation functions of both models coincide.  To summarize the result, let us introduce the notation $\mean{A;B;\cdots}$ to denote the {\em cumulants\/} of random variables $A,B,\ldots$;
for instance, $\mean{A;B}=\mean{AB}-\mean{A}\mean{B}$.
\begin{prop}[\cite{DKL}]
\label{prop:DKL}
We have
\begin{equation*}
\nessvf{A ; B} = \nesssc{A ; B} \, ,
\end{equation*}
for any two random variables $A$ and $B$ linear in $\{q_1, \ldots, q_N, p_1, \ldots, p_N\}$.
\end{prop}

It follows from Proposition \ref{prop:DKL} that both models satisfy Fourier's law with the same value of the conductivity. Indeed, since the energy current
\begin{equation*}
j_{x}^e = -\cfrac{1}{2} (p_x +p_{x-1}) (q_{x} -q_{x-1} ), \qquad x=2, \ldots, N-1 \, ,
\end{equation*}
is a quadratic function of the momenta and positions, its averages over $\nesssc{\cdot}$ and $\nessvf{\cdot}$ are equal. This implies that the conductivity of the self-consistent model, $\kappa_{\rm{s,s}}(T)$, and the conductivity of the velocity-flip model, $\kappa_{\rm{s,v}}(T)$, both defined by
\begin{equation*}
\kappa_{\rm{s,\cdot}}(T)= \lim_{T_{L}, T_r \to T} \, \lim_{N \to \infty}\, \cfrac{N \langle{j_x^e}\rangle_{\rm{s}, \cdot } }{T_\ell -T_r}
\end{equation*}
are equal to the same value $\kappa (T)$. The existence of these limits have been proved in \cite{BLL04} for the pinned ($\nu >0$) self-consistent chain and in \cite{BO2} for the unpinned ($\nu =0$) velocity-flip model. By using Proposition \ref{prop:DKL}, we get the existence and equality of these limits for both models in the pinned and unpinned cases.  The value of $\kappa (T)$ is given by
\begin{equation}
\label{eq:K}
\kappa (T) =\cfrac{1/\gamma}{2+\nu^2 + \sqrt{\nu^2 (\nu^2 +4)}}.
\end{equation}

\newcommand{\covmeas}{\rho_{\text{s,cov}}}

The steady state $\nesssc{\cdot}$ of the self-consistent chain is Gaussian. The probability measure
$\nessvf{\cdot}$ is not Gaussian but we will show it is a mixture of Gaussian states.
\begin{prop}
\label{prop:ss}
There exists a probability measure $\covmeas$ whose support is included in the set $\Sigma$ composed of positive definite symmetric matrices, such that $\nessvf{\cdot}$ is given by
\begin{equation*}
\nessvf{\cdot} = \int G_C (\cdot)  \, \covmeas(dC) .
\end{equation*}
\end{prop}

\begin{proof}
See Appendix \ref{sec:ness-mixture}.
\end{proof}

There is a simple non-trivial consequence of this property:
\begin{cor} Let $f(\bq, \bp)$ be any function of the form
$$f(\bq,\bp)=\sum_{i=1}^N  \left(a_i q_i + b_i p_i\right) $$
where $(a_1,\ldots,a_N, b_1, \ldots,b_N)$ are arbitrary real numbers. Then, the following inequality holds:
\begin{equation*}
\nessvf{f^4} \ge 3 \left[ \nessvf{f^2} \right]^2 = \nesssc{f^4}.
\end{equation*}
\end{cor}
\begin{proof}
If $Z \in \RR^d$, $d \ge 1$, is a centered Gaussian random variable then ${\bb E} (Z^4) = 3 \left[ {\bb E} (Z^2) \right]^2$. Thus, for any Gaussian measure $G_C$ with covariance matrix $C$, we have $G_C (f^4) =3 \left[ G_C (f^2) \right]^2$. By Proposition \ref{prop:ss} and Jensen's inequality, the result follows.
\end{proof}

\section{Hydrodynamical scaling limit of the velocity flip model}
\label{sec:hl}

We have to distinguish two cases according to whether $\nu=0$ (unpinned) or $\nu>0$ (pinned). The unpinned case is similar to that investigated by one of the authors in \cite{Ber}.

\subsection{The unpinned chain}
\label{subsec:uc}

When $\nu=0$ it is useful to define the more convenient coordinates $r_x, \; 0 \le x \le {N}$ by  $r_x =q_{x+1} -q_x$, and set $r_{N+1}=0$. The phase space is then identified (with some abuse of notation) as $\Omega_N= {\RR}^{N+1} \times {\RR}^{N}$ and a configuration $(r_0, \ldots, r_N, p_1, \ldots,p_N)$ is denoted by $\omega=(\br, \bp)$. The configuration at time $t$ of the process is denoted by $\omega (t)$.

In terms of these new variables the Liouville operator can be written as
\begin{equation*}
  \mc A = \sum_{x=0}^{N}\left(p_{x+1} - p_{x} \right) \partial_{r_x}  +
 \sum_{x=1}^{N}  \left( r_{x} -r_{x-1} \right)  \partial_{p_x}
\end{equation*}
The {\textit{bulk}} dynamics conserves two quantities. The first one is the energy ${\mc H}$. The second one is the deformation $\sum_{x} r_x$ of the lattice. Because of the Langevin baths the total energy is no longer conserved by the {\textit{total}} dynamics; the situation is different for the total deformation  $\sum_{x=0}^{N} r_x$ since our boundary conditions $q_0=q_{N+1}=0$ fix it to $0$ at any time. Nevertheless, the field $\br(t)=\{r_x (t)\, ; \, 0 \le x \le N \}$ fluctuates, and this second conservation law has to be taken into account in the hydrodynamic analysis. Observe that the conservation of total momentum is no longer valid because it is broken by the noise.

%

The energy at a site $x \in \{1, \ldots, N\}$ is now given by
\begin{align}\label{eq:defunpex}
{\mc E}_x = \cfrac{p_x^2}{2} +\cfrac{r_x^2}{4} + \cfrac{r_{x-1}^2}{4}
\end{align}
and ${\mc E}_0 = \frac{1}{4} r_0^2\,$, ${\mc E}_{N+1} = \frac{1}{4} r_{N}^2\,$.
The local conservation of the two conserved quantities is expressed by the equations
\begin{equation}
\label{eq:lc}
\begin{split}
{\mc E}_x (t) -{\mc E}_x (0) = \int_0^{t} ds \; \left\{ j_x^e (s) -j_{x+1}^e (s)\right\}ds\\
r_x (t) - r_x (0) = \int_0^{t } ds \; \left\{ j_x^r (s) -j_{x+1}^r (s)\right\}ds\
\end{split}
\end{equation}
where the current of the energy $j_x^e$ and of the current of the deformation of the lattice $j_x^r$ are given by
\begin{equation}
\label{eq:cur}
j_x^e = -\cfrac{1}{2} r_{x-1} (p_x + p_{x-1}), \quad j_x^r = -p_x \, .
\end{equation}

The equilibrium Gibbs measures of the infinite system are parametrized by the inverse temperature $\beta=T^{-1}$ and the pressure (or tension) $\tau$. They are given by
\begin{equation}
\mu_{T,\tau} (\omega) = Z(T,\tau)^{-1} \prod_{x=0}^{N+1} \exp\left\{ -\beta ({\mc E}_x - \tau r_x) \right\}
\end{equation}
and sometimes denoted by $\langle \cdot\rangle_{T,\tau}$. We note that the following relations hold for $x \in \{1, \ldots, N\}$
\begin{equation}
\langle p_x^{2} \rangle_{T,\tau} =T, \quad \langle {\mc E}_x \rangle_{T,\tau} = T + \tau^2 /2, \quad \langle r_x \rangle_{T,\tau} =\tau \, .
\end{equation}

There is a second natural parameterization of the equilibrium measures by the averages of the two conserved quantities. Let $E,R$ be the functions defined by
\begin{equation}
\label{eq:cor}
E(T, \tau) = T +\tau^2 /2, \quad R (T, \tau) = \tau .
\end{equation}
For any ${\bar R} \in \RR$, ${\bar E} > \frac{1}{2} {\bar R}^2$, we denote by ${\hat \mu}_{\bar E, \bar R}$ the canonical equilibrium measure $\mu_{T,\tau}$ with $\tau,T$ such that $E(T,\tau)=\bar E$, $R(T,\tau) ={\bar R}$.  Explicitly, these are given by $\tau={\bar R}$ and $T={\bar E} - \frac{1}{2} {\bar R}^2>0$.

To establish the hydrodynamic limits corresponding to the two conservation laws, we look at  the process with generator $N^{2} {\mc L}$, i.e., in the diffusive scale. Assume that initially the process is started with a Gibbs local equilibrium measure ${\hat \mu}_{\varepsilon_0 (\cdot), u_0 (\cdot)}$ associated with a macroscopic deformation profile $u_0 (q)$ and a macroscopic energy profile ${\varepsilon}_0 (q)$:
\begin{align}\label{eq:locG}
& {\hat \mu}_{\varepsilon_0 (\cdot), u_0 (\cdot)}   \\  &  \quad \nonumber
=
\cfrac{1}{Z (T_0 (\cdot), \tau_0 (\cdot))}
\prod_{x=0}^{N+1} \exp \left\{ -\beta_0 (x/N) ({\mc E}_x - \tau_{0} (x/N) r_x ) \right\}
 dr_x dp_x
\end{align}
where $T_0 = \beta_0^{-1}$ and $\tau_0$ are the temperature and tension profiles corresponding to the given energy and deformation profiles through the correspondence defined in (\ref{eq:cor}). The corresponding expectation is denoted by $\langle \cdot \rangle$.  Both profiles are assumed to be continuous and then
we have for any macroscopic point $q\in [0,1]$
\begin{equation}
\lim_{N \to \infty} \langle r_{[Nq]} (0) \rangle = u_0(q), \quad  \lim_{N \to \infty} \langle {\mc E}_{[Nq]} (0) \rangle = {\varepsilon}_0 (q) \, .
\end{equation}

We show that at any later (macroscopic) time $t$
\begin{equation}
\lim_{N \to \infty} \langle r_{[Nq]} (N^2 t) \rangle = u(q,t), \quad  \lim_{N \to \infty} \langle {\mc E}_{[Nq]} (N^2 t) \rangle = {\varepsilon} (q,t) \, ,
\end{equation}
where $u, \varepsilon$ are solutions of the following macroscopic diffusion equation
\begin{equation}
\label{eq:he}
\begin{cases}
\partial_t u = \gamma^{-1}\,  \partial_q^2 \, u\\
\partial_t {\varepsilon}  = (2 \gamma)^{-1} \, \partial_q^2 \, ( \varepsilon + u^2 /2)
\end{cases}
\end{equation}
with the initial conditions $u(\cdot, 0) =u_0 (\cdot)$, $\varepsilon (\cdot, 0) = {\varepsilon} _0 (\cdot)$ and boundary conditions
\begin{equation}
\label{eq:eqhlboundaries}
\begin{split}
&(\partial_q u) \, (0,t)=(\partial_q u) \, (1,t)=0,\\
&{\left(\ve -\cfrac{u^2}{2}\right) (0,t) = T_\ell, \quad \left(\ve -\cfrac{u^2}{2}\right) (1,t)= T_r\, .}
\end{split}
\end{equation}

Taking the limit $t \to \infty$ in these equations we obtain the typical macroscopic profiles of the system in the nonequilibrium stationary state $\langle \cdot \rangle_{ss}$, i.e., a flat deformation profile $u=0$ and a linear profile ${\bar T}$ interpolating between $T_\ell$ and $T_r$,
\begin{equation}
\label{eq:bT}
{\bar T} (q) =T_\ell + (T_r -T_\ell) q \, ,
\end{equation}
for the energy profile.

\subsubsection{Derivation of Macroscopic Equations}

 We sketch a proof of (\ref{eq:he})  at a somewhat informal level. It is based on the following two part analysis. A fully rigorous treatment would involve dealing with several non-trivial technical problems. We will not do this here but see \cite{Ber0,BO} for a rigorous proof in a similar context.

\begin{enumerate}
 \item {\it Local thermal equilibrium (LTE):\/}
Local thermal equilibrium corresponds to each given small macroscopic region of the system being in equilibrium, but different regions may be in different equilibrium states, corresponding to different values of the parameters (see (\ref{eq:locG})). Let us consider a very large number, say  $2L+1$, of sites in microscopic units ($L\gg 1$), but still an infinitesimal number at the macroscopic level ($(2L+1)/N \ll 1$). We choose hence $L=\delta N$ where $\delta \ll 1$. We consider the system in the box $\Lambda_L (x)$ composed of the sites labeled by $x-L, \ldots, x+L$. The time evolution of the $2L+1$ oscillators is essentially given by the bulk dynamics since the variations of deformation and energy in the volume containing the $2L+1$ oscillators changes only via boundary fluxes. After $N^2$ microscopic time units, the system composed of the $2L+1$ oscillators has relaxed to the micro-canonical state $\lambda_{{\bar e}_q (t
 N^2), {\bar r}_q (t N^2)}$ at $q=x/N$ corresponding to the local empirical deformation ${\bar r}_q (tN^2)$ and the local empirical energy ${\bar e}_q (t N^2)$ in the box $\Lambda_L (x)$. Indeed, we can divide the observables into two classes, according to their relaxation times: the fast observables, which relax to equilibrium values on a time scale much shorter than $N^2$ and will not have any effect on the hydrodynamical scales, and the slow observables which are locally conserved by the dynamics and need much longer times to relax. Let $\phi (\omega)$ be a local observable and let $\phi_x = \tau_x \phi$ denote the same observable translated by $x$. Thermal local equilibrium corresponds to $\langle\phi_{[Nq]} (N^2 t)\rangle$ being close to its value in $\lambda_{{\bar e}_q (tN^2), {\bar r}_q (tN^2)}$. In the thermodynamic limit we take first $N \to \infty$ and then $\delta \to 0$ (we recall that $\delta$ is related to the size of the box $\Lambda_L (x)$ by $L=\delta N$),
 we have, by equivalence of ensembles,
\begin{equation*}
\langle\phi_{[Nq]} (N^2 t)\rangle \approx {\hat \mu}_{{\varepsilon} (q,t), {u} (q,t)} (\phi)
\end{equation*}
since
$${\bar e}_q (tN^2) \approx {\varepsilon} (q,t), \; {\bar r}_q (tN^2) \approx {u} (q, t)\, .$$
\item {\it First order corrections to LTE:\/}  Since we observe the system on a diffusive scale, local thermal    equilibrium alone is not sufficient to determine the form of the hydrodynamic equations. The second key ingredient is the computation of first order corrections to LTE. These corrections are expressed (in the bulk) by what is called a ``fluctuation-dissipation equation" in the hydrodynamic limit mathematical literature.  In general, such equations are not explicit, but this happens for our system. We begin by noting that the currents are given, on the microscopic level, by
\begin{equation}
\label{eq:fd}
\begin{split}
&j_x^e = - (\nabla \phi)_x +{\mc L} (h_x), \quad j_x^r = -\gamma^{-1} (\nabla r)_{x-1} + {\mc L} (\gamma^{-1} p_x), \\
&\phi_x= (2 \gamma)^{-1} (p_{x-1}^2 +r_{x-1} r_{x-2}), \quad h_x = -\gamma^{-1} j_{x}^e \, ,
\end{split}
\end{equation}
where ${\mc L}$ is defined in (\ref{eq:generatorL}) and $(\nabla\phi)_x = \phi_{x+1}-\phi_x$.
\end{enumerate}

To get the hydrodynamical equations, we let $G : [0,1] \to \RR$ be a test function. Then we have
\begin{eqnarray*}
&& N^{-1} \sum_{x=1}^{N} G(x/N) \langle {\mc E}_{x} (tN^2) \rangle - N^{-1} \sum_{x=1}^{N} G(x/N) \langle {\mc E}_{x} (0) \rangle\\
&\approx& N^{-2} \sum_{x=1}^N \int_0^{tN^2} G' (x/N) \langle j_{x}^e (s) \rangle ds \\
&\approx& N^{-3} \sum_{x=1}^N \int_0^{tN^2} G'' (x/N) \langle \phi_{x} (s) \rangle ds \\
&& \quad + N^{-2} \sum_{x=1}^N \int_0^{tN^2} G' (x/N) \langle ({\mc L} h_x) (s) \rangle ds\\
&\approx& N^{-1} \sum_{x=1}^N \int_0^{t} G'' (x/N) \langle \phi_{x} (s N^2) \rangle ds  \\
&& \quad + N^{-2} \sum_{x=1}^N G' (x/N) \left\{ \langle h_x (tN^2) \rangle -  \langle h_x (0) \rangle\right\}\\
&\approx& N^{-1} \sum_{x=1}^N \int_0^{t} G'' (x/N) \langle \phi_{x} (s N^2) \rangle ds + {\mc O} (N^{-1})\, .
\end{eqnarray*}
To get the first equality, we have used local conservation of energy (\ref{eq:lc}) and a discrete integration by parts.
To get the second one, we used
(\ref{eq:fd}) and a second discrete integration by parts. To get the third equality,we used  Lemma \ref{lem:sto} in Appendix B.

By LTE the term $\langle \phi_{x} (s N^2) \rangle $ is asymptotically equivalent in the thermodynamic limit to
$${\hat \mu}_{{\varepsilon} (q,s), {r} (q,s)} (\phi_x)$$
where $q = x/N$. A simple computation shows that
$$\mu_{T,\tau} (\phi_x) = (2 \gamma)^{-1} \left( \tau^2 +T \right) = (2 \gamma)^{-1} \left[ E(T,\tau) + R^2 (T,\tau) /2 \right]$$
so that
$${\hat \mu}_{{\bar E}, {\bar R}} (\phi_x)=(2 \gamma)^{-1} \left( {\bar E} +{\bar R}^2 /2 \right) \, .$$
This closes the evolution equation for the energy field.

A similar computation can be performed for the deformation field giving the form of the hydrodynamic equations (\ref{eq:he}) in the bulk. These equations are, as already noted, to be solved subject to the boundary conditions (\ref{eq:eqhlboundaries}).
The reason for this is that the time relaxation of the Langevin baths is of order one and we look at the process at time $tN^2$. On this time scale the system has reached LTE with temperature $T_\ell$ at the left and temperature $T_r$ at the right. Since the average of the kinetic energy at (macroscopic) time $t$ and position $q \in [0,1]$ is given by $(\ve -u^2 /2) (q,t)$, we get the second equation in (\ref{eq:eqhlboundaries}). The Neumann boundary conditions for the $u$-field are a consequence of the conservation of the total length of the chain $\sum_{x=0}^N r_x (t) = q_{N+1} (t) -q_{0} (t)$ which remains equal to $0$ during the time evolution. Thus we have $\int_{[0,1]} u(q,t) dq =0$ for any time $t$. Since in the bulk $u$ is a solution of the heat equation, by taking the time derivative of the previous equality, we get $f(t):=\partial_q u (0,t)=\partial_q u (1,t)$ (the fluxes at the boundaries compensate). We have to show that $f(t)=0$. Observe that
\begin{equation}
\label{eq:eqsss}
\begin{split}
&{\mc L} p_1 = (r_{1} -r_{0}) - (\gamma +1) p_1, \quad {\mc L} r_0 =p_1.
\end{split}
\end{equation}
Thus, by Lemma \ref{lem:sto} in Appendix \ref{sec:asc}, we have for any $t>0$,
\begin{equation*}
\begin{split}
\cfrac{1}{N} \left\{ \langle p_{1} (tN^2)\rangle -\langle p_1 (0) \rangle \right\}&= \int_{0}^t \; N \left[ \langle r_1 (sN^2) \rangle -\langle r_0 (sN^2) \rangle \right]\,  ds\\
&-(\gamma +1)\, N \int_0^t \langle p_1(sN^2) \rangle ds,\\
\int_{0}^t \langle p_1 (sN^2) \rangle ds& = \cfrac{1}{N^2} \left\{ \langle r_0 (tN^2) \rangle - \langle r_0 (0) \rangle\right\}
\end{split}
\end{equation*}
The second equality shows that $N \int_0^t \langle p_1(sN^2) \rangle ds$ vanishes as $N$ goes to infinity. The left hand side of the first equality is of order $1/N$ and since $\int_{0}^t \; N \left[ r_1 (sN^2) -r_0 (sN^2) \right]\,  ds $ converges to $\int_0^t (\partial_q u)(0,s)ds = \int_0^t f(s) ds$, we have $f(s)=0$ for any $s$.

%

\subsection{The pinned chain}
\label{subsec:pc}
When $\nu>0$, the energy current $j_x^e$ satisfying ${\mc A} {\mc E}_x = j_{x}^e - j_{x+1}^e$  is given by
\begin{equation*}
j_{x}^e = -\cfrac{1}{2} (p_x +p_{x-1}) (q_{x} -q_{x-1} ) \, .
\end{equation*}
Because of the presence of the pinning, the bulk dynamics conserves only one quantity: the energy ${\mc H}$. It follows that the Gibbs equilibrium measures of the infinite system $\mu^{\rm{eq}}_{T}$ are fully characterized by the temperature $T=\beta^{-1}$. Under $\mu^{\rm{eq}}_T$ the $p_x$ are independent of the $q_x$ and are independent identically distributed Gaussian variables of variance $T$. The $q_x$ are distributed according to a centered Gaussian process with covariance $\mu_T (q_x q_y) = \Gamma (x-y)$ such that
\begin{equation}
\label{eq:Gamma}
[\nu^2 -\Delta] \Gamma \, (z) = T\, {\bf 1}_{\{0\}} (z) \, .
\end{equation}
In particular, we have $\mu_T (p_x^2) =\mu_T ({\mc E}_x) =T$.

The ``fluctuation-dissipation" equation (\ref{eq:fd}), which gives the first order corrections to local thermal equilibrium, is now
\begin{equation}
\label{eq:fd22}
j_x^e = -(\nabla \phi)_x + {\mc L} (h_x)
\end{equation}
with $h_x = -\gamma^{-1} j_x^e$, as before, but we need to choose now
\begin{align}
\label{eq:phi007}
&\phi_x = \cfrac{1}{2\gamma} \left[ p_{x-1}^2 - (\nu+1) q_{x-1}^2 -q_x q_{x-2} +q_{x-1}q_x +q_{x-1}q_{x-2}\right]\, .
\end{align}

Assume that the system is initially distributed according to a Gibbs local equilibrium measure associated to the energy profile ${\varepsilon}_0 (q), \; q \in [0,1]$, and define $\varepsilon (q,t)$ as the evolved profile in the diffusive scale, i.e.,
\begin{equation}
{\ve } (q,t) = \lim_{N \to \infty} \langle {\mc E}_{[Nq]} (t N^2) \rangle \, .
\end{equation}
Then the arguments of Subsection \ref{subsec:uc} and (\ref{eq:fd22}) show that $\varepsilon$ is the solution of the following heat equation
\begin{equation}
\begin{cases}
\partial_t \varepsilon = \partial_q (D (\varepsilon)  \partial_q {\varepsilon})\, ,\\
\varepsilon (\cdot, 0) = \varepsilon_0 (\cdot)\, ,\\
\varepsilon (0,t)= T_\ell, \; \varepsilon (1,t) =T_r \, .
\end{cases}
\end{equation}
The diffusivity (in general called the conductivity in the context of heat conduction) $D(T)$ is given by
\begin{equation}
D (T) = -\cfrac{\mu_T (\phi_x)}{T}
\end{equation}

By (\ref{eq:phi007}) and (\ref{eq:Gamma}) we get that
\begin{eqnarray*}
\mu_T (\phi_x) &=& \cfrac{1}{2\gamma} \left[ T -(\nu^2 +1) \Gamma (0) -\Gamma (2) +2 \Gamma (1) \right]\\
&=& \cfrac{1}{2\gamma} \left[ T - \nu^2 \Gamma (0) - (\Delta \Gamma )(1) \right]\\
&=& \cfrac{1}{2\gamma} \left[ T - \nu^2 (\Gamma (0) +\Gamma (1))\right]
\end{eqnarray*}
and $(\nu^2 +2) \Gamma (0) -2\Gamma (1) =T$, $\Gamma (0) =T \int_0^1 (\nu^2 + a \sin^{2} (\pi k))^{-1} dk$. Thus we have
\begin{equation}
\label{eq:kappa}
D:= D(T) =   \cfrac{1/\gamma}{2+\nu^2 + \sqrt{\nu^2 (\nu^2 +4)}}
\end{equation}
Observe that, as expected, the value of $D$ is equal to the value of the conductivity $\kappa (T)$ in the NESS (see (\ref{eq:K})).

When $t$ goes to infinity ${\varepsilon} (q,t)$ converges to the linear profile ${{\bar T}} (q)= T_\ell + (T_r- T_\ell)q$. The latter is the typical profile observed in the stationary state $\langle \cdot \rangle_{{\rm{s,v}}}$.
We note finally that, since the self-consistent model does not conserve energy in the bulk, we do not expect any autonomous macroscopic equations in that model.

\section{Energy Fluctuations in the stationary state of the velocity flip model}
\label{sec:cor}

\subsection{The unpinned chain}
\label{sec:unpin}
In order to obtain the form of the energy fluctuations in the steady state we first derive the form of the energy {\textit{dynamical fluctuations}} in the diffusive scale. We start again with the unpinned case.

\subsubsection{Dynamical fluctuations for the unpinned case}
\label{sec:dynfluc}

The deformation and energy fluctuation fields are defined by
\begin{equation*}
\begin{split}
{\mc R}_t^N (F) = \cfrac{1}{\sqrt N} \sum_{x=1}^N F(x/N) \left[ r_x (tN^2) - u(x/N, t) \right],\\
{\mc Y}_{t}^N (G) = \cfrac{1}{\sqrt N} \sum_{x=1}^N G(x/N) \left[ {\mc E}_x (tN^2) - {\varepsilon}(x/N, t) \right],
\end{split}
\end{equation*}
where $F,G : [0,1] \to \RR$ are smooth test functions vanishing at the boundaries and $u,\varepsilon$ are the solutions of the hydrodynamic equations (\ref{eq:he}). For notational convenience we also introduce the temperature profile $T(q,t)$ associated to $\varepsilon, u$ and defined by
\begin{equation}
T(q, t)= \varepsilon (q,t) - \cfrac{u^2 (q,t)}{2}
\end{equation}

 We assume that the process is distributed according to a Gibbs local equilibrium state with deformation profile $u_0$ and energy profile $\varepsilon_0$. One shows easily that as $N$ goes to infinity $({\mc R}_0^N, {\mc Y}_0^N)$ converges to some limit field $({\mc R}_0, {\mc Y}_0)$ but since we do not need the explicit form of the latter in the sequel we do not compute it here (see \cite{KL}).

Let $a,b,c$ be the space-time dependent functions given by
\begin{equation}
\label{eq:abc}
a= u \sqrt {\cfrac{2T}{\gamma}},\quad b= \cfrac{T}{\sqrt \gamma},\quad c= \sqrt {\cfrac{2T}{\gamma}}.
\end{equation}
We claim that as $N$ goes to infinity $({\mc R}^N, {\mc Y}^N)$ converges to $({\mc R}, {\mc Y})$ given by the solution of the following coupled equations
\begin{equation}
\label{eq:tf}
\begin{cases}
\partial_t {\mc R} = \gamma^{-1} \, \partial_q^2\,  {\mc R} - \, \partial_q \,  \left( c W_{1}\right)\\
\partial_t {\mc Y} = (2 \gamma)^{-1} \left( \, \partial_q^2 \,  (u {\mc R}) +\, \partial_q^2 \,  {\mc Y} \right) -\, \partial_q \, \left( a W_1 + b W_2 \right)\, ,
\end{cases}
\end{equation}
where $W_1, W_2$ are two independent space-time white noises. The initial conditions, independent of $W_1,W_2$, are given by ${\mc R}_0$ and ${\mc Y}_0$.

The proof of (\ref{eq:tf}) is the following. For any  Gibbs equilibrium measure  ${\hat \mu}_{{\bar E},{\bar R}}$ with mean energy ${\bar E}$ and mean deformation ${\bar R}$, and for any local observable $f$ depending on the configuration $\omega$,  let ${\hat f} ({\bar E}, {\bar R})={\hat \mu}_{\bar E,\bar R}(f)$. Using the fluctuation dissipation equations (\ref{eq:fd}), standard stochastic calculus and the hydrodynamic equations (\ref{eq:he}), it is not difficult to show that
\begin{equation}
\label{eq:324}
\begin{cases}
{\mc R}_t^N (F) -{\mc R}_0^N (F) \approx \gamma^{-1} \int_0^t {\mc R}_s^N (F'') ds + {\mc M}_{1}^N (t), \\
{\mc Y}_t^N (G) -{\mc Y}_0^N (G) \approx  \int_0^t  \cfrac{1}{\sqrt N} \sum_x G'' (x/N) \left[ \phi_x (sN^2)  -{\hat \phi}_x (\ve (\frac{x}{N}, s), u(\frac{x}{N},s)) \right] ds\\
\phantom{{\mc Y}_t^N (G) -{\mc Y}_0^N (G) \approx } + {\mc M}_{2}^N (t)\
\end{cases}
\end{equation}
where
\begin{equation}
\label{eq:hatphi}
{\hat \phi_x} ({\bar E},{\bar R})= (2\gamma)^{-1} \left ( {\bar E} + {\bar R}^2 /2 \right)
\end{equation}
and ${\mc M}_1^N, {\mc M}_2^N $ are two martingales. By using (\ref{eq:416}) of Appendix \ref{sec:asc} and the LTE assumption, we get that, as $N$ goes to infinity,
\begin{equation*}
\begin{split}
& \langle  \left[ {\mc M}_1^N (t) \right]^2  \rangle  \to \cfrac{2}{\gamma} \int_0^t ds \int_0^1 dq [H' (q)]^2 T(q,s), \\
& \langle  \left[ {\mc M}_2^N (t) \right]^2  \rangle \to \cfrac{1}{\gamma} \int_0^t ds \int_0^1 dq [G' (q)]^2 \, T(q,s) \left( T(q,s) +2 u^2 (q,s)\right), \\
& \langle {\mc M}_1^N (t) {\mc M}_2^N (t)  \rangle \to \cfrac{2}{\gamma} \int_0^t ds \int_0^1 dq [H' (q) G' (q)] \, T(q,s) u(q,s).
\end{split}
\end{equation*}
This implies that the martingales $({\mc M}_1^N, {\mc M}_2^N)$ converge as $N$ goes to infinity to $(c W_1 (H') , (a W_1 +bW_2) (G'))$ with $W_1, W_2$ two standard independent space-time white noise and $a,b,c$ are the functions introduced in (\ref{eq:abc}).

The first equation of (\ref{eq:324}) is closed and the martingale convergence result gives the first equation of (\ref{eq:tf}).
The second one is not and the closure is obtained through a {\textit{nonequilibrium Boltzmann-Gibbs principle}} \cite{LOM} which generalizes the standard {\textit{equilibrium Boltzmann-Gibbs principle}} \cite{BR}. Let us explain what it means. Observables are divided into two classes: non-hydrodynamical and hydrodynamical. The first ones are the
non-conserved quantities and they fluctuate on a much faster scale than the conserved ones. Hence, they should average out and only their projection on the hydrodynamical variables should persist in the scaling limit. Thus there exist macroscopic space-time dependent functions $C(q,s), D(q,s)$ such that for every test function $J(q)$, $q \in [0,1]$, and any local observable $f$ depending on the configuration $\omega$,
\begin{equation*}
\begin{split}
&\lim_{N \to \infty} \cfrac{1}{\sqrt N} \int_0^t ds \sum_{x=1}^{N} J (x/N) \left\{ (\tau_x f )( sN^2  ) -{\hat f}  (\ve ({x}/{N}, s), u({x}/{N},s))  \right. \\
&\quad \quad \quad \quad \quad \quad  \left. -C(x/N, s) \, (r_x -u(x/N,s )) -D (x/N, s)\, \left({\mc E}_x-\varepsilon (x/N,s )\right)\right\}\\
&=0
\end{split}
\end{equation*}
in mean square norm. The functions $C$ and $D$ depend on the macroscopic point $q=x/N$ and on the macroscopic time $s$. In order to compute them, we assume local thermal equilibrium. Around the macroscopic point $q$, the system is considered at equilibrium with a fixed value of the deformation $u(s,q)$ and of the temperature $T(s,q)$. The functions $C,D$ are then computed by projecting the function $\tau_x f$ on the deformation and energy fields. The values of $C$ and $D$ are given by
\begin{equation*}
C(q,s)=( \partial_{\bar R} {\hat f})\,  (u(q,s), {\varepsilon} (q,s)), \quad D (q,s)= (\partial_{\bar E} {\hat f})  \,(u(q,s),{\varepsilon} (q,s)).
\end{equation*}
Using the Boltzmann-Gibbs principle we see that we can close the time evolution equation for the energy fluctuations field. By (\ref{eq:hatphi}), we have
$$C(q,s)= \cfrac{u(q,s)}{2\gamma}, \quad D (q,s)= \cfrac{1}{2\gamma}$$
and we obtain the second equation in (\ref{eq:tf}).

\subsubsection{Stationary fluctuations}

We are now ready to compute the energy and deformation fluctuations in the steady state, i.e., in the large $N$ limit of
\begin{align*}
& {\mc R}_{ss}^N (F) = \cfrac{1}{\sqrt N} \sum_{x=1}^N F (x/N)  r_x,\\
& {\mc Y}_{ss}^N (G) = \cfrac{1}{\sqrt N} \sum_{x=1}^N G(x/N) \left[ {\mc E}_x  - {\bar T} (x/N) \right],
\end{align*}
where the configuration $(\br, \bp)$ is distributed according to the steady state $\langle \cdot \rangle_{{\rm{s,v}}}$ and $\bar T$ is the temperature profile (\ref{eq:bT}).

We show in Appendix \ref{sec:ap} that {$({\mc Y}_{ss}, {\mc R}_{ss})= \lim_{N \to \infty} \, ({\mc Y}_{ss}^N, {\mc R}_{ss}^N)$ are two independent Gaussian fields with covariance given by
\begin{equation}
\label{eq:sf}
\begin{cases}
\nessvf{ \left[ {\mc R}_{ss} (F) \right]^2} =  \int_{0}^1  {\bar T} (q) F^{2} (q) dq \\
\\
\nessvf{ \left [ {\mc Y}_{ss} (G) \right ]^2} = \int_{0}^1 {\bar T}^2 (q) G^{2} (q) dq +(T_r -T_\ell)^2 \int_{0}^1 G (q) ((-\Delta_0)^{-1} G) (q) dq\\
\end{cases}
\end{equation}
where $\Delta_0 $ denotes the Laplacian with Dirichlet boundary conditions on $[0,1]$.

\subsection{The pinned chain}\label{sec:pinnedLDF}

Our goal is to estimate the probability that in the stationary state the empirical energy profile defined by
\begin{equation}
\theta^N (q) = \sum_{x=1}^{N} {\mc E}_x {\bf 1}_{[x/N, (x+1)/N)} (q), \quad q \in [0,1]
\end{equation}
is close to a prescribed macroscopic energy profile $e (q)$ different from ${{\bar T}} (q)$, i.e. we want to find the large deviation function (LDF) for the NESS.

At equilibrium, $T_\ell = T_r=T=\beta^{-1}$, the stationary state $\langle \cdot \rangle_{ss}$ coincides with the usual Gibbs equilibrium measure $\mu^{N, {\rm eq}}_T$ defined by
\begin{equation*}
d \mu^{N, {\rm eq}}_T = \cfrac{1}{Z_N (T)} \exp\left( - \beta \sum_{x=1}^{N} {\mc E}_x\right) \, \prod_{x=1}^N dq_x dp_x
\end{equation*}
where $Z_N (T)$ is the partition function. By the usual large deviations theory (see e.g. \cite{DZ}) we have that for any given macroscopic energy profile $\ve (\cdot)$
\begin{equation}
\label{eq:ld1}
{\mu}^{N,{\rm eq}}_T \left(  \theta^{N} (\cdot) \sim \ve (\cdot) \right)  \sim e^{-N V_{\rm{eq}} (e)}
\end{equation}
where
\begin{equation}
V_{\rm{eq}} (e) =\int_0^1 \left\{ \cfrac{e(q)}{T} -1- \log\left(\cfrac{e(q)}{T}\right) \right\} dq \, .
\end{equation}
The quantity $V_{\rm{eq}}$ is the large deviation function (LDF) in the canonical ensemble. The LDF coincides for equilibrium systems with the difference between the free energy of the system in LTE  and the equilibrium free energy.  It will thus vanish when $e(q)=T$ .





The purpose of the macroscopic fluctuation theory (MFT) of Bertini et al.\ \cite{Bjsp} is to obtain the probability of a large deviation for boundary driven diffusive systems. It predicts that out of equilibrium ($T_\ell \ne T_r$) a large deviation principle in the following form holds
\begin{equation}
\mu_{\rm{s,v}}^N \left( \theta^{N} (\cdot) \sim e (\cdot) \right)  \sim e^{-N V(e)} \, .
\end{equation}
Moreover, the rate function $V$ can be identified with the so-called quasi-potential, a quantity introduced in the context of stochastically perturbed dynamical systems \cite{FW}, and generalized by Bertini et al.\ to the infinite dimensional context; $V$ is then the LDF for the energy profile.

The explicit form of the quasi-potential is in general unknown. But the MFT has been rigorously proven in the context of the boundary driven simple exclusion process \cite{B51,BG,F} and formally derived for few other systems \cite{Ber,Breview,BGL}.
The main interest of the MFT is that it does not rely on the microscopic properties of the system studied but only on two macroscopic quantities: the diffusivity and the mobility \cite{Breview}. In particular, the LDF $V$ depends only on these two data.

For the system we are interested in, the diffusivity $D(T)$ is given by (\ref{eq:kappa}). By the Einstein relation the mobility $\chi (T)$ is equal to $\chi (T)= D(T) \sigma (T)$ where $\sigma (T)$ is the static compressibility defined by
\begin{equation}
\sigma (T) := \sum_{x \in \ZZ} \mu_T \left( ({\mc E}_0 -T) ({\mc E}_x -T) \right) \, .
\end{equation}
A simple computation shows that $\sigma (T) =T^2$ and Theorem 6.5 of \cite{Breview} applies.
It follows that $V(\cdot)$ is given by
\begin{equation}
V(e)=\int_0^1 dq \left[ \cfrac{e(q)}{F(q)} -1 -\log\left( \cfrac{e(q)}{F(q)}\right) - \log \left( \cfrac{F' (q)}{T_r -T_\ell }\right) \right]\, ,
\end{equation}
where $F$ is the unique increasing solution of
\begin{equation}
\begin{cases}
\cfrac{\partial^2_q F}{(\partial_q F)^2 } = \cfrac{F-e}{F^2}\, ,\\
F(0)=T_\ell , \; F(1) =T_r\, .
\end{cases}
\end{equation}
Thus the function $V$ is independent of the pinning value $\nu^2$ and of the intensity of the noise $\gamma$. In fact, it coincides with the LDF of the KMP model considered in \cite{BGL}.

It is now easy to derive the Gaussian fluctuations of the empirical energy. We consider a small perturbation, $e={{\bar T}} + \delta h$, of the stationary profile ${\bar e}$. The functional $V$ has a minimum at ${{\bar T}}$ so that
\begin{equation*}
V(e) = V({{\bar T}}) + \cfrac{1}{2} \, \delta^2 \, \langle h, C^{-1} h \rangle +o (\delta^2)
\end{equation*}
The operator $C$ is the covariance for the Gaussian fluctuations of the empirical energy under the invariant measure $\mu_{\rm{s,v}}^N$. The computations are exactly the same as in \cite{BGL}, section 5.2, and we get
\begin{equation}
C= {{\bar T}}^2 {\bf 1} + (T_r -T_\ell)^2 (-\Delta_0)^{-1}
\end{equation}
where $\Delta_0$ denotes the Laplacian with Dirichlet boundary conditions on $[0,1]$.  Therefore, it coincides with the expression (\ref{eq:sf}) even though (\ref{eq:sf}) is derived assuming $\nu=0$.

\subsection{Higher dimensions}

It should be possible to extend the previous results to higher dimensions, with or without pinning, apart from the derivation of the LDF performed in subsection \ref{sec:pinnedLDF}.  A possible difficulty for $\nu =0$ is that the Gibbs equilibrium measure has long range correlations in the $q$ variables and they are even unbounded in the volume for $d=1,2$.  At equilibrium, these problems can be avoided by considering the corresponding gradient fields \cite{gos01}.  For instance, in $d=1$ these are given by the variables $r_x = q_{x+1}-q_x$, which have a product equilibrium measure.


\newcommand{\Tsc}{T^{\rm sc}}

\section{Energy correlations in the steady state for the self-consistent chain}
\label{sec:corsc}

In this section we consider the self-consistent chain and mainly follow the notations and results in \cite{BLL04}.
We show that the energy correlations in the steady state are not long range in the sense that
\begin{equation}
\label{eq:sclr}
\frac{1}{N} \nesssc{{\mc H};{\mc H}} =  \frac{1}{N} \sum_{j',j=0}^{N+1}
\eqmean{{\mc E}_{j'} ;    {\mc E}_j}{\Tsc_j}
 + R_N \, ,
\end{equation}
where $\{\Tsc_j\}$ denotes the self-consistent profile and $R_N$ is a remainder of order $1/\sqrt{N}$ in the pinned case and $N^{-1/4} \ln^2 N$ in the unpinned case. The fluctuations are thus dominated by the local equilibrium term, whose profile is well approximated by a linear profile connecting the fixed boundary temperatures

\subsection{The pinned self-consistent chain} \label{sec:pinnedscc}

We begin with the pinned case $\nu>0$, since the necessary estimates in this case have been proven in \cite{BLL04}.
The local energy at site $j \in \{0, \ldots,N+1\}$ can be written as
%
\begin{align}
 & {\mc E}_j = \frac{1}{2} p_j^2 + \frac{1}{2} \bq^T\! A^{(j)} \bq\, , \quad
  \text{with } \\ & \quad A^{(j)}_{xy} = \nu^2 \cf[j=x=y] + \cf[x=y] \sum_{y'=0}^{N+1} B^{(j)}_{xy'} - B^{(j)}_{xy}\, ,
 \nonumber \\ & \quad
 B^{(j)}_{xy} = \cf[|x-y|=1] \frac{1}{2} (\cf[x=j]+\cf[y=j])\, ,\nonumber
\end{align}
where $\cf[P]=1$, if the condition $P$ is true, and $\cf[P]=0$ otherwise.
For any choice of temperatures of the heat baths $\bT = \{ T_j\, ; \, j=1, \ldots,N\}$ the NESS $\muss ^{N, \bT}$, denoted to simplify notations by $\langle \cdot \rangle$, is Gaussian with mean zero. Hence we have
\begin{align}\label{eq:ejejex}
 & 4 \mean{{\mc E}_{j'} {\mc E}_j}
  \\ & \quad = \mean{p_{j'}^2 p_{j}^2} + \mean{\bq^T\!
   A^{(j)} \bq\, \bq^T\! A^{(j')} \bq}  +  \mean{p_{j'}^2 \bq^T\! A^{(j)} \bq} +\mean{p_{j}^2 \bq^T\!
   A^{(j')} \bq}
 \nonumber \\ & \quad
 = 4 \mean{{\mc E}_{j'}} \mean{{\mc E}_j} +  2\mean{p_{j'} p_{j}}^2 +
  2\sum_{x'y'xy} A^{(j')}_{x'y'}A^{(j)}_{xy}\mean{q_{x'}q_x}
  \mean{q_{y'}q_y}
    \nonumber \\ & \qquad
 + 2 \sum_{xy} A^{(j')}_{xy} \mean{p_j q_x} \mean{p_j q_y}
 + 2 \sum_{xy} A^{(j)}_{xy} \mean{p_{j'} q_x} \mean{p_{j'} q_y}
 \, ,  \nonumber
\end{align}
where the sums go only over $\{1,2,\ldots,N\}$, as $0=q_0=q_{N+1}=p_0=p_{N+1}$.
In particular, at thermal equilibrium with all the temperatures $T_j$ equal to $T$, we have
\begin{align}
 & 2 \eqmean{{\mc E}_{j'} ; {\mc E}_j}{T}
 = T^2 \Bigl( \cf[j'=j\in\{1,2,\ldots,N\}]   \\ & \qquad \nonumber
  + \mathrm{Tr}\! \left[A^{(j')} (-\Delta+\nu^2)^{-1}
   A^{(j)}(-\Delta+\nu^2)^{-1}\right]\Bigr) \, .
\end{align}
Since $\nu>0$, $(-\Delta+\nu^2)^{-1}_{xy}$ is exponentially
decreasing in $|x-y|$ with an exponent independent of $N$.
Therefore, using the fact that $A^{(j)}_{xy}$ is zero unless
$|x-j|,|y-j|\le 1$,
\begin{align}
  \sum_{j'=0}^{N+1} \left|\mathrm{Tr}\! \left[A^{(j')} (-\Delta+\nu^2)^{-1}
      A^{(j)}(-\Delta+\nu^2)^{-1}\right]\right|
\end{align}
is  uniformly bounded in $N$ and $j$.

Let $\bT=\{ \Tsc_j \, ; \, j=1, \ldots,N\}$ denote the self-consistent temperature profile, which has been
proven to lie between the boundary temperatures $T_\ell$ and $T_r$, and thus is
uniformly bounded in $N$.  For convenience, define also $T_0=T_\ell$ and $T_{N+1}=T_r$.  Then,
\begin{align}\label{eq:EEeqlb}
\sum_{j',j=0}^{N+1} \left|\eqmean{{\mc E}_{j'}; {\mc E}_j}{T_j} \right| = O(N)\, .
\end{align}
Thus the fluctuations of the energy at the NESS satisfy
\begin{align}
 & \frac{1}{N} \mean{{\mc H};{\mc H}} = \frac{1}{N} \sum_{j',j=0}^{N+1}
 \mean{{\mc E}_{j'}; {\mc E}_j}
 = \frac{1}{N} \sum_{j',j=0}^{N+1} \eqmean{{\mc E}_{j'};{\mc E}_j}{T_j}
 + R_N\, ,
\end{align}
where the first term is $O(1)$ and,
by (\ref{eq:ejejex}), the remainder is bounded by
\begin{align}
 & |R_N| \le \frac{1}{2 N} \sum_{j',j=0}^{N+1} \Bigl|
  \delta_j(\mean{p_{j'} p_{j}}^2) +
  \sum_{x'y'xy} A^{(j')}_{x'y'}A^{(j)}_{xy}\delta_j(\mean{q_{x'}q_x}
  \mean{q_{y'}q_y})
    \\ & \qquad \nonumber
 + \sum_{xy} A^{(j')}_{xy} \delta_j(\mean{p_j q_x} \mean{p_j q_y})
 +  \sum_{xy} A^{(j)}_{xy} \delta_j(\mean{p_{j'} q_x} \mean{p_{j'}
   q_y} )\Bigr| \, .
\end{align}
Here we have used the shorthand notation $\delta_j$ to denote the difference
between a quantity evaluated with NESS and the equilibrium distribution at
temperature $T_j$, i.e.,
\begin{align}
  \delta_j \mean{X} & = \mean{X} - \eqmean{X}{T_j},\quad\text{and}\\
  \delta_j(\mean{X} \mean{Y}) &= \mean{X}
  \mean{Y}-\eqmean{X}{T_j}\eqmean{Y}{T_j}\, .
\end{align}
Since
\begin{align}
 \delta_j(\mean{X} \mean{Y})
  = \mean{Y} \delta_j \mean{X}+ \eqmean{X}{T_j}\delta_j \mean{Y}\, ,
\end{align}
we have
 \begin{align} \label{eq:RNbound}
 & |R_N| \le \frac{1}{2 N}
 \sum_{j=1}^N T_j \delta_j(\mean{p_{j}^2})
  \\ & \quad +
 \frac{1}{2 N} \sum_{j',j=0}^{N+1}
 \Bigl| \sum_{x'y'xy}  A^{(j')}_{x'y'}A^{(j)}_{xy}\eqmean{q_{x'}q_x}{T_j}
  \delta_j(\mean{q_{y'}q_y})\Bigr|
    \nonumber \\ & \quad
 +
 \frac{1}{2 N} \sum_{j',j=0}^{N+1} \Bigl|\mean{p_{j'} p_{j}}
  \delta_j(\mean{p_{j'} p_{j}}) +
  \sum_{x'y'xy} A^{(j')}_{x'y'}A^{(j)}_{xy}\mean{q_{x'}q_x}
  \delta_j(\mean{q_{y'}q_y})
    \nonumber \\ & \qquad
 +\sum_{xy} A^{(j')}_{xy} \mean{p_j q_y}\delta_j(\mean{p_j q_x} )
 + \sum_{xy} A^{(j)}_{xy} \mean{p_{j'} q_y} \delta_j(\mean{p_{j'}
   q_x} )\Bigr| \, .  \nonumber
\end{align}

It is proven in \cite{BLL04} that $T_j$ is linear up to corrections which are
uniformly bounded by $N^{-1/2}$,
and that $|\delta_j(\mean{X Y})|\le C N^{-1/2}$ for any choice of
$X,Y\in \{q_j,p_j\}_j$ where the constant $C$ depends only on the fixed boundary
temperatures and on $\nu,\gamma$. In particular, it was proven that
$\vep_N:= \max_{1\le k<N} | T_{k+1}-T_k|=O(N^{-1/2})$.
In fact, as we shall show next, the previous estimates can be
straightforwardly strengthened to an upper bound with exponential decay in the
difference $|j'-j|$.

We use the results proven in Sec.~2--4 of \cite{BLL04}.  There it is shown
that to each of the pair $(X,Y)$ in $(q,q)$,
$(q,p)$, $(p,q)$, $(p,p)$ corresponds a $B^{(k)}_{xy}$ which is independent of
the temperatures of the heat baths and using which
$\mean{X_{x} Y_{y}}=\sum_{k=1}^N B^{(k)}_{xy} T_k$ in the NESS of any
temperature profile $\{T_k>0\}$.
In addition, it was shown that
\begin{align}\label{eq:bkxy}
   B^{(k)}_{xy} & = \hat{f}_N(x-k,y-k) + \hat{f}_N(x+k,y+k)
      \\ & \quad  \nonumber
  - \hat{f}_N(x-k,y+k) - \hat{f}_N(x+k,y-k)\, ,
\end{align}
where the function $\hat{f}_N$ is exponentially decaying in both arguments in
the precise sense that there are $N$-independent constants $C$, $\alpha>0$
such that
\begin{align}\label{eq:fnexpbound}
  |\hat{f}_N(x,y)| \le C \rme^{-\alpha (|x|'+|y|')}\, ,
\end{align}
with the periodic distance $|x|':= |x \bmod 2(N{+}1)|$ where $x \bmod 2(N{+}1)
\in \{-N,-N+1,\ldots,N+1\}$.
For $x,r\in \{1,2,\ldots,N\}$, one obviously has $|x-r|\le N-1$ implying
$|x-r|' = |x-r|$. But $|x+r|'=x+r$ only for $x+r\le N+1$, and else
$|x+r|'=2 N +2 -x-r$.  For $x+r> N+1$ thus
$|x+r|'= 2(N+1-x)+x-r = 2 (N+1-r)+r-x$ which implies that $|x+r|'\ge |x-r|$,
and this inequality is obviously true for $x+r\le N+1$.  Thus we can conclude
that for any choice of the signs
\begin{align}\label{eq:persum}
  |x\pm r|'+ |y \pm r|'
\ge |x-r|+|y-r| \ge \frac{1}{2} (|x-y| + |x-r| +|y-r|)\, ,
\end{align}
where the second inequality follows from the triangle inequality.  Therefore,
we have now proven that
\begin{align}\label{eq:Brxybound}
   |B^{(r)}_{xy}| \le 4 C \rme^{-\frac{1}{2}\alpha (|x-y| + |x-r| +|y-r|)}\, .
\end{align}

Using these estimates, we can thus improve Corollary 4.2 of \cite{BLL04} to
conclude that for any temperature profile
and for any $x,y,j$ in $\{1,2,\ldots,N\}$,
\begin{align}
& \left|\mean{X_{x} Y_{y}}-\eqmean{X_{x} Y_{y}}{T_j}\right| \le \sum_{k=1}^N
|B^{(k)}_{xy}| |T_k-T_j|
    \\ & \quad
\le 4 C \sum_{k=1}^N \rme^{-\frac{1}{2}\alpha (|x-y| + |x-k| +|y-k|)} |k-j| \vep_N
    \nonumber \\ & \quad \nonumber
    \le 4 C  \vep_N\rme^{-\frac{1}{2}\alpha |x-y|}
     \sum_{n\in \Z} \rme^{-\frac{1}{2}\alpha |n|}(|n|+|x-j|) \, ,
\end{align}
where $\vep_N:= \max_{1\le k<N} |T_{k+1}-T_k|$.  Clearly, the final bound is then valid also for $j=0,N+1$ simply since $T_0=T_1$ and $T_{N+1}=T_N$.
Thus for the self-consistent profile there is $C'>0$ independent of $N$ such
that
\begin{align}\label{eq:ltediff}
& \left|\delta_j(\mean{X_{x} Y_{y}})\right|
\le \frac{C'}{\sqrt{N}} (1+|x-j|) \rme^{-\frac{1}{2}\alpha |x-y|}\,.
\end{align}
In addition, (\ref{eq:Brxybound}) implies that there is $C''>0$ such that for
any $T>0$,
\begin{align}
& \left|\eqmean{X_{x} Y_{y}}{T}\right| \le T \sum_{k=1}^N  |B^{(k)}_{xy}| \le
C'' T\rme^{-\frac{1}{2}\alpha |x-y|}\, ,
\end{align}
and for the selfconsistent NESS
\begin{align}
& \left|\mean{X_{x} Y_{y}}\right| \le C'' \max(T_\ell,T_r)
\rme^{-\frac{1}{2}\alpha |x-y|}\, .
\end{align}

Applying these bounds in (\ref{eq:RNbound}) shows that there is $C$
independent of $N$ such that
\begin{align}
&  |R_N|\le \frac{C}{\sqrt{N}} \Bigl[ 1 + \frac{1}{N} \sum_{|x-y|\le 1,
  |x'-y'|\le 1}
  \rme^{-\frac{1}{2}\alpha (|x'-x|+|y'-y|)}
  \\ & \quad \nonumber
  + \frac{1}{N} \sum_{j}  \sum_{|x-y|\le 1}
  \rme^{-\frac{1}{2}\alpha (|j-x|+|j-y|)}
   + \frac{1}{N} \sum_{j'}  \sum_{|x-y|\le 1}
  \rme^{-\frac{1}{2}\alpha (|j'-x|+|j'-y|)}
 \Bigr]\, .
\end{align}
Clearly, then $R_N=O(N^{-1/2})$, and thus it becomes negligible compared to
the local equilibrium part when $N\to\infty$.



\subsection{The unpinned self-consistent chain} \label{sec:scmunpinned}

The unpinned self-consistent chain with $\nu=0$ was not considered in \cite{BLL04}, mainly because the decay of correlations is no longer exponential which complicates the analysis.  It is not completely straightforward to see that the decay is sufficiently strong for the previous argument to work, and one has to consider the right subset of observables to find strong enough decay; even at equilibrium $\mean{q_j q_{j'}}= T ((-\Delta)^{-1})_{j,j'}$ which implies that
$\sum_{j,j'}\mean{q_j q_{j'}}= O(N^3)$.  However the energy only depends on the $r_x$ variables, and these turn out to have better decay properties. We show in this section that (\ref{eq:sclr}) holds also in the self consistent profile for $\nu=0$: its first term is $O(1)$ and the remainder $R_N$ has a bound
$O(N^{-1/4} \ln^2 N)$ and is thus dominated by the first, local equilibrium, term.

Even with $\nu=0$, the formulae for the steady state covariance given in Section 2 of \cite{BLL04} hold, and each component is a continuous function of $\nu\ge 0$.  Thus we can make a shortcut by considering the $\nu\to 0$ limits of the earlier derived expressions for the steady state expectations.  Explicitly, we get from Sections 2 and 4 of \cite{BLL04} that for any $i,j\in\{1,2,\ldots,N\}$, $\nu>0$, and any temperature distribution $\{T_n\}$, there is a unique steady state with covariance matrix satisfying
\begin{align}\label{eq:defqxcorr}
 & \mean{q_i q_j} = \sum_{n=1}^N B^{(n)}_{1,\nu}(i,j) T_n\,, \quad
 \mean{q_i p_j} = \sum_{n=1}^N B^{(n)}_{2,\nu}(i,j) T_n \\ \nonumber & 
  \mean{p_i p_j} = \sum_{n=1}^N B^{(n)}_{3,\nu}(i,j) T_n\, ,
\end{align}
with the explicit definitions of $B^{(n)}_{a,\nu}$ given in Appendix \ref{sec:scmunpcorr}.  We recall the definition of $r_x=q_{x+1}-q_x$ given in Section \ref{subsec:uc}, as well as the representation ${\mc E}_x = \frac{p_x^2}{2} +\frac{r_x^2}{4} + \frac{r_{x-1}^2}{4}$ of the local energies in terms of these variables.
Therefore, we have
\begin{align*}
 & {\mc E}_j = \frac{1}{2} p_j^2 + \frac{1}{2} {\bf r}^T\! \tilde{A}^{(j)} {\bf r}\, , \quad
  \text{with } 
\tilde{A}^{(j)}_{xy} = \cf[x=y] \frac{1}{2} (\cf[x=j]+\cf[x=j-1])\, ,\nonumber
\end{align*}
and thus in order to study the local energy correlations it suffices to consider the correlations of variables $r$ and $p$, following the computations in the previous subsection.  We also point out, that since these variables are independent in the equilibrium Gibbs measure, we have $\eqmean{{\mc E}_{j'}; {\mc E}_j}{T}=0$ unless
$|j'-j|\le 1$.

The main new ingredient allowing us to extend the result of \cite{BLL04} to the unpinned case are the bounds proven in Appendix \ref{sec:scmunpcorr}.
There we show that the exponential decay of the pinned correlations will be replaced by a powerlaw decay, whose strength depends on the observable.  Explicitly, we prove there that in addition to (\ref{eq:defqxcorr}) also
\begin{align}\label{eq:defrxcorr}
 & \mean{r_x r_y} = \sum_{n=1}^N B^{(n)}_{5,\nu}(x,y) T_n\,, \quad
 \mean{r_x p_j} = \sum_{n=1}^N B^{(n)}_{6,\nu}(x,j) T_n \, ,
\end{align}
and the relevant $B$-matrices decay in $m:=1+|x-y|+|x-n|+|y-n|$ so that there is a constant $c_0$ independent of $L,x,y,n$ such that for $\nu=0$
\begin{align}\label{eq:rxcorrdecay}
 & \left| B^{(n)}_{5,\nu}(x,y)\right| \le c_0 m^{-2} (1+\ln m) \, ,\quad
 \left| B^{(n)}_{6,\nu}(x,y)\right| \le c_0 m^{-3}(1+\ln m) \, ,\\ \nonumber &
 \left| B^{(n)}_{3,\nu}(x,y)\right| \le c_0 m^{-4} \, .
\end{align}

Instead of going through all the details here, let us only list the changes necessary in the argument used in \cite{BLL04}.  As mentioned before, for any temperature profile there is a unique Gaussian steady state, and the explicit formulae for its covariance matrix are still valid for $\nu=0$.  In addition, Section 3 of  \cite{BLL04} also holds verbatim, since the key bound in Lemma 3.2 does not use the pinning.  This shows that also in the unpinned case there is a unique self-consistent temperature profile which is bounded between the fixed boundary temperatures.  The main changes needed will be to Section 4, since the above bounds produce an additional logarithmic term in the local equilibrium approximations.  By the bounds in (\ref{eq:rxcorrdecay}) we have for $\nu=0$ and $a=3,5,6$ that there is a constant $c'$ such that
\begin{align}
  \sum_{n=1}^N |B^{(n)}_{a,\nu}(x,y)| \le c' (1+|x-y|)^{-\frac{3}{4}}\begin{cases}
                                                     1, &\text{ if } a=5\, , \\
                                                     (1+|x-y|)^{-1}, &\text{ if } a=6\, , \\
                                                     (1+|x-y|)^{-2}, &\text{ if } a=3\, , \\
                                                   \end{cases}
\end{align}
(the power is not optimal but will be sufficient here) and
\begin{align}
  \sum_{n=1}^N |x-n| |B^{(n)}_{a,\nu}(x,y)| \le c' \ln^2 N\begin{cases}
                                                     1, &\text{ if } a=5\, , \\
                                                     (1+|x-y|)^{-1}, &\text{ if } a=6\, , \\
                                                     (1+|x-y|)^{-2}, &\text{ if } a=3\, . \\
                                                   \end{cases}
\end{align}
Thus, in particular,
for $X_x,Y_x\in \{r_x,p_x\}$ the equilibrium and the self-consistent expectations satisfy
\begin{align*}
& \left|\eqmean{X_{x} Y_{y}}{T}\right| \le  C T (1+|x-y|)^{-\frac{3}{4}-b}\, , \quad \\
& \left|\mean{X_{x} Y_{y}}\right| \le C \max(T_\ell,T_r) (1+|x-y|)^{-\frac{3}{4}-b}\, ,
\end{align*}
where $b=0,1,2$, depending on the choice of $X$ and $Y$.
This proves that $\eqmean{{\mc E}_{j'}; {\mc E}_j}{T}$ is uniformly bounded in $N$, and since it is zero unless $|j'-j|\le 1$ we can conclude that (\ref{eq:EEeqlb}) holds in the selfconsistent profile also for $\nu=0$.
Finally, denoting $\vep_N:= \max_{n} | T_{n+1}-T_n|$ we have
\begin{align}
& \left|\mean{X_{x} Y_{y}}-\eqmean{X_{x} Y_{y}}{T_j}\right|
\le C \vep_N (1+|x-y|)^{-b} (1+\ln^2 N +|x-j|)\, .
\end{align}

We can now use these bounds to complete the proof of Fourier's law and derive an estimate of $\vep_N$ for the selfconsistent profile, following section 4.3 of \cite{BLL04}.  There it was shown that the total current only depends on the expectation $\mean{q_1 q_1}-\mean{q_N q_N}$.  Since this is equal to $\mean{r_0 r_0}-\mean{r_N r_N}$, we can conclude from the above estimates that the discussion on page 795 of \cite{BLL04} holds verbatim, and therefore also for $\nu=0$ we have $\vep_N=O(N^{-1/2})$ in the selfconsistent profile.  (There will be an additional $\ln^2 N$ term multiplying the right hand side of (4.15) in \cite{BLL04}, but this will not change the conclusions.)

We can now follow the same steps as in Section \ref{sec:pinnedscc} and conclude that the bound in (\ref{eq:RNbound}) for the remainder holds, after we change $A$ to $\tilde{A}$ and $q$ to $r$ on its right hand side.  Then the above estimates can be applied there, proving that the term is bounded by $N^{-1/4} \ln^2 N$.  The worst term is the fourth one which only has decay $|x'-x|^{-3/4}$ arising from the NESS expectation $\mean{r_{x'} r_x}$, and thus leads to a bound $\vep_N N^{1/4} \ln^2 N$.  This concludes the proof of the claims made in the beginning of this section.

\subsection{Higher dimensions}

The easiest way to extend the selfconsistent system to a square lattice in higher dimensions is to label the particles with 
${\bf x}\in \{1,2,\ldots,N\}^d=: I_N^d$, and then use fixed boundary conditions in the first direction, as before, and periodic boundary conditions for the harmonic dynamics in the other directions.  
We have shown in \cite{BLL04} that this system has a unique selfconsistent profile if all heat baths attached to ${\bf x}$ with $x_1=1$ are fixed to temperature $T_\ell$ and those with $x_1=N$ are fixed to temperature $T_r$.  Moreover, this profile only depends on $x_1$, and the 
different Fourier-modes in the orthogonal directions decouple.  

Explicitly, if we define for ${\bf k}\in I_N^{d-1}$, $j\in I_N$,
\begin{align}
 \widehat{q}_j({\bf k}) := N^{-\frac{d-1}{2}} \sum_{{\bf y}\in I_N^{d-1}} q_{(j,{\bf y})} \rme^{-\ci 2\pi {\bf y}\cdot{\bf k}/N}\, , 
\end{align}
and $\widehat{p}_j({\bf k})$ similarly, then in the selfconsistent steady state
\begin{align}
 \mean{\widehat{q}_{i'}({\bf k'})^* \widehat{p}_i({\bf k)}} = \mean{q_{i'} p_{i}}^{(1;{\bf k})} \cf({\bf k'}={\bf k}) \, ,
\end{align}
where the second mean is taken in the stationary state of the {\em one-dimensional\/} selfconsistent chain with pinning 
$\tilde\nu^2({\bf k}):= \nu^2 + 4 \sum_{i=1}^{d-1} \sin^2(\pi k_i/N)\ge \nu^2$.  Analogous equations are valid for 
$\mean{\widehat{q}_{i'}({\bf k'})^* \widehat{q}_i({\bf k)}}$ and 
$\mean{\widehat{p}_{i'}({\bf k'})^* \widehat{p}_i({\bf k)}}$.

Using the above partially Fourier-transformed variables, we can write the total energy as
\begin{align}
 {\mc H} = \sum_{i=0}^{N+1} \sum_{{\bf k}\in I_N^{d-1}} \Bigl[ |\widehat{p}_i({\bf k)}|^2 + \tilde\nu^2({\bf k}) |\widehat{q}_i({\bf k)}|^2
 + \frac{1}{4} \sum_{|i'-i|=1} |\widehat{q}_{i'}({\bf k)}-\widehat{q}_i({\bf k)} |^2 \Bigr]\, .
\end{align}
Therefore,
\begin{align}
 \nesssc{{\mc H}} =  \sum_{{\bf k}\in I_N^{d-1}}\sum_{i=0}^{N+1}\mean{{\mc E}_{i}({\bf k})}^{(1;{\bf k})} \, ,
\end{align}
where
\begin{align}
{\mc E}_{i}({\bf k}) := \cfrac{p_i^2}{2} + \tilde{\nu}({\bf k})^2 \cfrac{q_i^2}{2} + \cfrac{1}{4} \sum_{i':\, |i'-i|=1} (q_{i'} -q_i)^2\, .
\end{align}
In addition, using the independence of the modes and the property $\widehat{q}_i({\bf k})^*=\widehat{q}_i(-{\bf k})$, we also find
\begin{align}\label{eq:higherdefl}
 \frac{1}{N^d} \nesssc{{\mc H};{\mc H}} = \frac{1}{N^{d-1}} \sum_{{\bf k}\in I_N^{d-1}} \frac{1}{N} \sum_{i',i=0}^{N+1} 
\mean{{\mc E}_{i'}({\bf k}) ; {\mc E}_i({\bf k})}^{(1;{\bf k})} \, .
\end{align}

Therefore, the energy fluctuations are given by a convex combination of the previous one-dimensional fluctuations, only with a varying pinning parameter, $\tilde{\nu}({\bf k})^2\ge \nu^2$.  It follows that the statements made in the beginning of this section generalize to higher dimensions, i.e., we can conclude that
\begin{equation}
 \frac{1}{N^d}\nesssc{{\mc H};{\mc H}} =
  \frac{1}{N^d} \sum_{{\bf x'},{\bf x}\in I_N^{d}}
\eqmean{{\mc E}_{{\bf x'}} ;    {\mc E}_{{\bf x}}}{\Tsc_{\bf x}}
 + R_N \, ,
\end{equation}
where the correction term to local equilibrium, $R_N$, is of order
$1/\sqrt{N}$ in the pinned case and $N^{-1/4} \ln^2 N$ in the unpinned case.  The second bound could likely be improved by using (\ref{eq:higherdefl}), but since neither of the one-dimensional bounds is optimal, let us skip such computations.

%
\vspace*{1em}

{\textsc{Acknowledgements} \; We thank S. Olla for useful discussions. C. Bernardin acknowledges the support of the French Ministry of Education through the ANR-10-BLAN 0108 grant.  J.~Lukkarinen was supported by the Academy of Finland. The authors thank the Rutgers Unversity, the Fields Institute, the IHES, the IHP, and NORDITA for their hospitality. This work was supported in part by NSF Grants DMR 08-02120 and by AFOSR Grant FA9550-10-1-0131.}

\appendix

\section{The NESS of the velocity-flip model}
\label{sec:ness-mixture}

Let us first consider the harmonic chain in contact with only two heat baths. Its generator is ${\mc A} + {\mc B}_{1,T_\ell} + {\mc B}_{N,T_r}$. It is easy to show that if we start the dynamics from an initial centered Gaussian state ${\tilde \mu}_0$ then the law of the process at time $t$ is given by a centered Gaussian state ${\tilde \mu}_t$. We shall denote by ${\tilde C} (t)$ the $2N \times 2N$ covariance matrix of ${\tilde \mu}_t$. It can be written in the form
\begin{equation*}
{\tilde C} (t) =
\left(
\begin{array}{cc}
{\tilde U} (t) & {\tilde Z} (t)\\
{\tilde Z}^* (t) & {\tilde V} (t)
\end{array}
\right)
\end{equation*}
where ${\tilde U}, {\tilde V}, {\tilde Z}$ are $N \times N$ matrices defined by
\begin{equation*}
{\tilde U}_{i,j} (t) = \int q_{i} q_{j} d{\tilde \mu}_t, \quad {\tilde V}_{i,j} (t) = \int p_{i} p_{j} d{\tilde \mu}_t, \quad {\tilde Z}_{i,j} (t) = \int q_{i} p_{j} d{\tilde \mu}_t .
\end{equation*}
Let $T= (T_\ell +T_r)/2$ and $2\eta= T^{-1} (T_\ell -T_r)$. We introduce the matrices $A$ and $D$ defined by
\begin{equation*}
A =
\left(
\begin{array}{cc}
0 & - I\\
\Phi & R
\end{array}
\right), \quad
D =
\left(
\begin{array}{cc}
0 & 0\\
0 & 2T (R +\eta S)
\end{array}
\right)
\end{equation*}
where the $N \times N$ matrices $R$, $S$ and $\Phi$ are given by $R_{i,j} = \delta_{i,j} (\delta_{i,1} + {\delta}_{i,N})$, $S_{i,j}= \delta_{i,j} ( \delta_{i,1} -\delta_{i,N})$ and $\Phi_{i,j} = (\nu^2 +2) \delta_{i,j} -\delta_{i+1,j} -\delta_{i-1,j}$

Then ${\tilde C} (t)$ is the solution of the differential equation (see \cite{RLL})
\begin{equation}
\label{eq:diffC}
\cfrac{d}{dt} \, {\tilde C} (t) = D -A{\tilde C} -{\tilde C} A^* .
\end{equation}
As $t$ goes to $+\infty$, ${\tilde C} (t)$ converges to a positive definite symmetric matrix ${\tilde C}_{\infty}$ satisfying $D=A {\tilde C}_{\infty} +{\tilde C}_{\infty} A^*$.  We note that ${\tilde C}_{\infty} =T C_{\rm{eq}} (1) + \eta {\hat C}$ where $\hat C$ is independent of $T_\ell, T_r$.

We now define a Markov process $(C(t))_{t \ge 0}$ with state space $\Sigma$, the set composed of the $2N \times 2N$ symmetric non-negative matrices $C$.
This Markov process is a kind of dual process of the velocity flip process.
For any $x \in \{1, \ldots, N\}$ we shall denote by $\Gamma_x$ the diagonal matrix of $\Sigma$ such that $(\Gamma_x)_{i,i}= 1- 2 \delta_{i,N+x}$. Consider $N$ independent Poisson processes  ${\mc N}_x$, $x=1, \ldots,N$, with rate $\gamma /2$ and denote by $\tau_{x} (1) < \ldots, \tau_{x} (k) < \ldots$ the successive times at which ${\mc N}_x$ jumps. Let $\tau_{x_1}^1 <\tau_{x_2}^2< \ldots$ be the sequence composed of all these times ordered in increasing order. The notation used is such that $\tau_{x_j}^j$ is a jump of the Poisson process ${\mc N}_{x_j}$, i.e., $\tau_{x_j}^j =\tau_{x_j} (k_j)$ for the $k_j$-th  jump of the site $x_j$.

The time evolution of $(C(t))_{t \ge 0}$ is given by the following rules: For $\tau_{x_j}^j \le t < \tau_{x_{j+1}}^{j+1}$ the evolution of $C(t)$ is deterministic and prescribed by (\ref{eq:diffC}). At time $t=\tau_{x_{j+1}}^{j+1}$, the value of $C(t)$ is given by $\Gamma_{x_{j+1}} C([\tau_{x_{j+1}}^{{j+1}}]^-) \Gamma_{x_{j+1}}$.
It is easy to check that the generator ${\hat{\mc L}}$ of $({C (t)})_{t \ge 0}$ is given, for any continuously differentiable function $F: \Sigma \to \RR$, by
\begin{equation}
\label{eq:gencapL}
\begin{split}
({\hat {\mc L}} F)(C) &= \sum_{x,y=1}^{2N} \left(D -AC - CA^* \right)_{x,y} \, \partial_{C_{x,y}} F \, (C)\\
&+ \cfrac{\gamma}{2} \, \sum_{x=1}^N \left[ F(\Gamma_x C \Gamma_x) -F(C) \right].
\end{split}
\end{equation}

The law of $(C(t))_{t \ge 0}$ when starting from $C_0$ is denoted by $\rho_t (C_{0}, dC)$ and the centered Gaussian measure on $\Omega_N$ with covariance matrix $C$ is denoted by $G_C$.
\begin{lemma}
Consider the velocity flip model and assume that the initial state is $\mu_0=G_{C_0}$. Then the law $\mu_t$ of the process at time $t$ is given
\begin{equation*}
\mu_t = \int_{\Sigma} G_C  \, \rho_t (C_0, dC).
\end{equation*}
\end{lemma}

\begin{proof}
The dynamics generated by ${\mc L}$ can be described as follows. Let $({\mc N}_x)_{x \in \{1, \ldots, N\}}$ be the sequence of independent Poisson processes on $(0,+\infty)$ introduced above. During the time interval $[\tau_{x_j}^j, \tau_{x_{j+1}}^{j+1})$ the process $(\omega (t))_{t \ge 0}$  follows the evolution prescribed by ${\mc A}+{\mc B}_{1,T_\ell} +{\mc B}_{1,T_r}$ and, at time $\tau_{x_{j+1}}^{j+1}$, the new configuration is obtained  by flipping the momentum of particle at site $x_{j+1}$. Then the system starts again following the dynamics generated by ${\mc A}+{\mc B}_{1,T_\ell} +{\mc B}_{1,T_r}$  until the time of the next flip, and so on.

Thus, conditionally to the realization of the $(\tau_{x_j}^j)_{j \ge 1}$, starting from $G_{C_0}$, the law of $\omega (t)$ at time $t$ is Gaussian with a covariance matrix $C(t)$ which is obtained by the following scheme: in the time interval $[ \tau_{x_j}^j, \tau_{x_{j+1}}^{j+1})$, $C(t)$ satisfied (\ref{eq:diffC}); at time $\tau_{x_{j+1}}^{j+1}$, the covariance is given by $\Gamma_{x_{j+1}} C( [\tau_{x_{j+1}}^{j+1}]^-) \Gamma_{x_{j+1}}$. This follows from the fact that if $\omega$ has law $G_{C}$ then $\omega^x$ has law $G_{C'}$ with $C' = \Gamma_x C \Gamma_x$.

Hence, the conditional law of the covariance matrix is given by the law of the process $(C(t))_{t\ge 0}$ conditionally to the realization of the Poisson processes. The result follows.
\end{proof}

We can now prove the following

\begin{prop}
There exists a probability measure $\covmeas$, which is invariant for $(C(t))_{t \ge 0}$ and  whose support is included in the set $\Sigma$ composed of definite positive symmetric matrices, such that $\nessvf{\cdot}$ is given by
\begin{equation*}
\nessvf{\cdot} = \int G_C \, \covmeas (dC)
\end{equation*}
\end{prop}

\begin{proof}
We assume $\nu>0$. The case $\nu=0$ can be treated similarly by considering the new variables $r_j = q_{j+1} -q_{j}$. The existence of an invariant probability measure for $(C(t))_{t \ge 0}$ is based on the existence of a good Lyapunov function.

For any  $\alpha>0$ we consider the function $W_{\alpha} : \Omega_n \to (0,+\infty)$ defined by
\begin{equation*}
W_{\alpha} (\omega)= \exp \left\{ \alpha {\mc H} (\omega) \right\}
\end{equation*}
It can be proved (see e.g.\ \cite{RB} or \cite{BO2}) that, if $\alpha>0$ is sufficiently small, then there exists a constant $K_0 >0$ such that for any initial condition $\omega_0 \in \Omega_n$,
\begin{equation*}
\sup_{t \ge 0} {\mathbb E}_{\omega_0} \left[  W_\alpha ({\tilde \omega} (t)) \right] \le K_0 \left( W_{\alpha} (\omega_0) +1 \right)
\end{equation*}
where $({\tilde \omega} (t))_{t \ge 0}$ is the process generated by ${\mc A} +{\mc B}_{1,T_\ell} +{\mc B}_{N,T_r}$. Since a flip of $p_x$ does not modify the value of the energy ${\mc H} (\omega)$, this inequality remains in force with ${\tilde \omega} (t)$ replaced by the velocity flip system ${\omega} (t)$:
\begin{equation}
\label{eq:boundlya}
\sup_{t \ge 0} {\mathbb E}_{\omega_0} \left[  W_\alpha ({\omega} (t)) \right] \le K_0 \left( W_{\alpha} (\omega_0) +1 \right)\, .
\end{equation}
By integrating $\omega_0$ with respect to $G_{C_0}$, we get
\begin{equation*}
\sup_{t \ge 0} \int \rho_t (C_0, dC) \, G_C \left[ W_{\alpha} \right] \le K_0 \left[ G_{C_0} \left[ W_{\alpha}\right] +1 \right]\, .
\end{equation*}

We choose a matrix $C_0$ such that the RHS of the previous inequality is finite and we remark that $W_{\alpha} \ge \alpha {\mc H}$ so that
\begin{equation*}
\sup_{t \ge 0} \int \rho_t (C_0, dC) \,  G_C \left[ {\mc H} \right] \le K'_0
\end{equation*}
for a suitable positive constant $K'_0$. Let us write the matrix $C$ as
\begin{equation*}
C=\left(
\begin{array}{cc}
{U} & {Z} \\
{Z}^* & {V}
\end{array}
\right)
\end{equation*}
and observe that since $C \in \Sigma$, $\max_{i,j} |C_{i,j}| \le \max_i \{U_{i,i}, V_{i,i} \}$, $U_{i,i}, V_{i,i} \ge 0$. We have  $G_C ({\mc H}) = {\rm tr} (U\Phi) + {\rm tr} (V)={\rm tr} (\Phi^{1/2} U {\Phi}^{1/2}) + {\rm tr} (V)$. Thus, there exists a constant $K''_0$ such that
\begin{equation*}
\sup_{t \ge 0} \int \rho_t (C_0, dC) \, \left[ \max_{i,j} |C_{i,j}|\right] \le K''_0\, .
\end{equation*}

This shows that the sequence of probability measures $\left\{ \rho_t (C_0, dC) \, ; \, t \ge 0 \right\}$ is tight. Consequently, any limiting point $\covmeas$ of the (tight) sequence
$$\left\{ {\bar \rho}_{t} (C_0, dC) := \cfrac{1}{t}  \int_0^t \rho_y (C_0, dC) \, dy  \, ; \, t \ge 0 \right\}$$
is an invariant probability measure of the process $(C(t))_{t \ge 0}$. Moreover, it is easy to show that $\int \covmeas (dC) \, G_C$ is equal to $\nessvf{\cdot}$. This follows from the fact that if $f:\Omega_n \to \RR$ is a bounded continuous function then $F: C \in \Sigma \to G_{C} (f)$ is a bounded continuous function on $\Sigma$, so that, if $t'$ is a sequence of times such that ${\bar \rho}_{t'}$ converges to $\covmeas$, we have
\begin{equation*}
\begin{split}
\nessvf{f}&= \lim_{t' \to \infty} \cfrac{1}{t'} \int_0^{t'} dy \, \mu_y (f)=\lim_{t' \to \infty} \cfrac{1}{t'} \int_0^{t'} dy \int \rho_y (C_0, dC)\,  G_C (f)\\
&= \lim_{t' \to \infty} {\bar \rho}_{t'} (F) = \covmeas (F)= \int_{\Omega_n} f(\omega) \, \left( \int \covmeas (dC) \, G_C (d\omega) \right) \, .
\end{split}
\end{equation*}
It has been proved in \cite{BO2} that $\nessvf{\cdot}$ has a density with respect to the Lebesgue measure. This implies that the support of $\covmeas$ is contained in the set $\Sigma$.
\end{proof}

\section{Stochastic calculus}
\label{sec:asc}
Here we give a proof of results we use to get in one of the steps involved in obtaining hydrodynamical equations below equation (\ref{eq:fd}) and the identities for the martingales following equation (\ref{eq:324})
We begin by stating the following standard lemma whose a proof can be found for example in \cite{KL} or \cite{RY}.
\begin{lemma}
\label{lem:sto}
If $A(\omega)$, $B (\omega)$ are two local functions of the configuration $\omega=(\br,\bp)$, then
 \begin{equation*}
 \begin{cases}
 A(\omega (tN^2)) - A(\omega (0))= N^2 \int_{0}^{t} ({\mc L}A) (\omega (s N^2)) ds +M^A_t \\
  B(\omega (tN^2)) - B(\omega (0))= N^2 \int_{0}^{t} ({\mc L} B) (\omega (s N^2)) ds +M^B_t
 \end{cases}
 \end{equation*}
where $M_t^A, M_t^B$ are two centered martingales with quadratic variations given by
\begin{align*}
\langle M_t^A M_t^B \rangle = \cfrac{\gamma N^2}{2} \sum_z \int_{0}^t  ds\,
\left[ A(\omega^z (sN^2) -A(\omega (sN^2) )\right] \\ \quad \times
\left[ B(\omega^z (sN^2) -B(\omega (sN^2) )\right] \, .
\end{align*}
\end{lemma}

The fluctuation-dissipation equations give
\begin{equation}
\begin{split}
&{\mc L} ({\mc E}_x) =\Delta \phi_{x+1} + {\mc L} (h_x -h_{x+1}),\\
&{\mc L} (r_x) = \gamma^{-1} \Delta r_x + \gamma^{-1} {\mc L} (p_x -p_{x+1})
\end{split}
\end{equation}
We apply Lemma \ref{lem:sto} with
\begin{equation}
\begin{split}
& A (\omega) = \cfrac{1}{\sqrt N} \sum_x F (x/N) (r_x + \gamma^{-1} \nabla p_x)\, ,\\
& B(\omega) = \cfrac{1}{\sqrt N} \sum_x G (x/N) ({\mc E}_x + \nabla h_x)\, .
\end{split}
\end{equation}
Let $F,G : [0,1] \to \RR$ be two test functions vanishing on the boundary.  Then
\begin{align}
& N^{-1/2} \sum_x F(x/N) (r_x +\gamma^{-1} \nabla p_x ) (tN^2) \\ \nonumber
&\qquad - N^{-1/2} \sum_x F(x/N) (r_x +\gamma^{-1} \nabla p_x ) (0)\\ \nonumber
& \quad \approx \gamma^{-1} \int_0^t N^{-1/2} \sum_x  F'' (x/N) r_x (sN^2) ds  +{\mc M}_1^N (t)
\end{align}
and
\begin{align}
& N^{-1/2} \sum_x G(x/N) ({\mc E}_x +\nabla h _x ) (tN^2) \\ \nonumber
&\qquad  - N^{-1/2} \sum_x G (x/N) ({\mc E}_x +\nabla h_x ) (0) \\ \nonumber
&\quad \approx \gamma^{-1} \int_0^t N^{-1/2} \sum_x  G'' (x/N) \phi_x (sN^2) ds  +{\mc M}_2^N (t)\, ,
\end{align}
where ${\mc M}^N_1, {\mc M}_2^N$ are two martingales whose quadratic variations are given by
 \begin{equation*}
 \begin{split}
& \langle M_1^N (t)  M_1^N (t)  \rangle = \cfrac{\gamma N^2}{2} \sum_z \int_{0}^t  ds \left[ A(\omega^z (sN^2) -A(\omega (sN^2) )\right] ^2\\
 &  \langle M_2^N (t)  M_2^N (t)  \rangle = \cfrac{\gamma N^2}{2} \sum_z \int_{0}^t  ds \left[ B (\omega^z (sN^2) -B (\omega (sN^2) )\right] ^2\\
 &\langle M_1^N (t)  M_2^N (t) \rangle = \cfrac{\gamma N^2 }{2} \sum_z \int_{0}^t ds \left[ A(\omega^z (sN^2) -A(\omega (sN^2) )\right] \\ \nonumber
&\qquad \times  \left[ B(\omega^z (sN^2) -B(\omega (sN^2) )\right]
 \end{split}
 \end{equation*}

 In fact, the two terms
 \begin{equation*}
 N^{-1/2} \sum_x F(x/N) (\nabla p_x ) (tN^2) - N^{-1/2} \sum_x F(x/N) (\nabla p_x ) (0)
 \end{equation*}
and
 \begin{equation*}
 N^{-1/2} \sum_x G(x/N) (\nabla h _x ) (tN^2) - N^{-1/2} \sum_x G (x/N) (\nabla h_x ) (0)
 \end{equation*}
 are small in $N$ (just perform a discrete integration by parts and use the smoothness of $F$ and $G$).
 Using again the smoothness of the functions $F$ and $G$ a simple computation shows that
 \begin{align}
 \label{eq:416}
 & \langle \left[ {\mc M}_1^N (t) \right]^2 \rangle \approx  \cfrac{2}{\gamma N} \sum_z \int_0^t (F')^2 (z/N) \langle p_{z}^2  (sN^2) \rangle ds \, ,
 \\ \nonumber &
 \langle \left[ {\mc M}_2^N (t) \right]^2 \rangle \approx  \cfrac{1}{2 \gamma N } \sum_z \int_0^t ds (G')^2 (z/N)
 \\ \nonumber &\quad \times
  \left\langle p_{z} (sN^2) \left(r_{z} (sN^2) +r_{z-1} (sN^2) \right)^2 \right\rangle  \, ,
 \\  \nonumber
   & \langle {\mc M}_1^N (t)  {\mc M}_2^N (t) \rangle \approx  \cfrac{1}{\gamma N } \sum_z \int_0^t ds (F') (z/N) (G') (z/N)
  \\ \nonumber &\quad \times\left\langle p_{z}^2  (sN^2) \left(r_{z} (sN^2) +r_{z-1} (sN^2) \right) \right\rangle  \, .
 \end{align}
This justifies the limits for the martingales following equation (\ref{eq:324}) listed in Section \ref{sec:dynfluc}.

 \section{Proof of equation (\ref{eq:sf})}
 \label{sec:ap}

 Since as $t \to \infty$, $u(q,t) \to 0$ and ${\varepsilon} (q,t) \to {\bar T} (q)$, $({\mc R}_{ss}, {\mc Y}_{ss})$ is the stationary solution of
\begin{equation*}
\begin{cases}
{\partial}_t {\mc R} = \gamma^{-1} \partial_q^2 {\mc R} -\partial_q \Bigl( \sqrt{\frac{2 {\bar T}}{\gamma}} W_1\Bigr)\, ,\\
\partial_t {\mc Y} = (2\gamma)^{-1} \partial_q^2 {\mc Y} -\partial_q \Bigl( \frac{\bar T}{\sqrt{\gamma}} W_2 \Bigr)\, .
\end{cases}
\end{equation*}
We prove the statement in (\ref{eq:sf}) only for the energy fluctuations field, the proof of the other equation being similar.

Let $G(q,t)=G_t (q)$ be the solution of
\begin{equation}
\begin{cases}
\partial_t G =(2\gamma)^{-1} \partial_q^2 G, \\
G(0,t)=G(1,t)=0,\\
G(q,0)=G_0 (q)\, .
\end{cases}
\end{equation}
Then  $G_t (q) = e^{\frac{t}{2\gamma} \Delta_0} G_0 (q)$, and we have
\begin{equation}
{\mc Y}_t (G_0)= {\mc Y}_{0} (G_t) +\frac{1}{\sqrt \gamma} \int_0^t \int_{[0,1]} {\bar T} (q) (\partial_q G_{t-s}) (q) W_2 (dq,ds) \, .
\end{equation}
As $t$ goes to infinity, $G_{t}$ goes to the zero function so that ${\mc Y}_0 (G_t)$ goes to $0$.   The second term converges to a Gaussian variable with variance equal to
\begin{align*}
& {\Sigma} (G_0)= \frac{1}{\gamma} \int_0^{\infty} ds \int_{[0,1]} dq\, {\bar T}^2 (q) (\partial_q G_s)^2 (q)\\
&\quad = -\frac{1}{\gamma} \int_0^{\infty} ds \int_{[0,1]} dq \left\{  G_s (q) \left( {\bar T}^2 (q) (\partial^2_q G_s)(q) + 2{\bar T} (q) {\bar T}' (q) (\partial_q G_s) (q) \right)\right\}\\
&\quad =-2\int_0^{\infty} ds \int_{[0,1]} dq {\bar T}^2 (q) G_s (q) (\partial_s G_s)(q)\\
&\qquad -\frac{2}{\gamma} (T_r -T_\ell) \int_0^{\infty} ds \int_{[0,1]} dq {\bar T} (q) (\partial_q G_s) (q) {G}_s (q)\\
&\quad = - \int_{[0,1]} dq {\bar T}^2 (q) \left\{ G_{\infty}^2 (q) -G_0^2 (q)  \right\} \\
&\qquad -\frac{1}{\gamma} (T_r -T_\ell) \int_0^{\infty} ds \int_{[0,1]} dq {\bar T} (q) (\partial_q G^2_s) (q) \\
&\quad = \int_{[0,1]} dq {\bar T}^2 (q) G_{0}^2 (q) + \frac{1}{\gamma} (T_r -T_\ell)^2  \int_0^{\infty} ds \int_{[0,1]} dq G_s^2 (q)
\end{align*}
Since $G_s = e^{\frac{s}{2\gamma} \Delta_0} G_0$, it follows that
\begin{align*}
&\int_0^{\infty} ds \int_{[0,1]} dq G_s^2 (q) = \int_0^{\infty} ds \int_{[0,1]} dq \left(e^{\frac{s}{2\gamma} \Delta_0} G_0 \right) (q) \, \left(e^{\frac{s}{2\gamma} \Delta_0} G_0 \right) (q)\\
&=  \int_0^{\infty} ds \int_{[0,1]} dq \left(e^{\frac{s}{\gamma} \Delta_0} G_0 \right) (q) G_0 (q)\\
&= \gamma \int_{[0,1]} dq G_{0} (q) ( (-\Delta_0)^{-1} G_0 ) (q) \, .
\end{align*}
This completes the proof of second equality in (\ref{eq:sf}).

\section{Decay of correlations in the unpinned selfconsistent model}\label{sec:scmunpcorr}

We begin from the expressions given in (\ref{eq:defqxcorr}).  By the results derived in Section 4 in \cite{BLL04}, we have there for $a=1,2,3$ and $\nu>0$
\begin{align}\label{eq:defbna}
 & B^{(n)}_{a,\nu}(i,j) \\ \nonumber & = g(i-n,j-n)+g(i+n,j+n)-g(i-n,j+n)-g(i+n,j-n)\, ,
\end{align}
with $g=g_{N,a,\nu}$ depending on $L=2 (N+1)$ and given by
\begin{align}\label{eq:defgnn}
  g(n',n) := \sum_{m',m\in \Z} \hat{f}_{a,\nu}(n' + m' L,n+m L)\, .
\end{align}
The function is obviously $L$-periodic, so may we assume that its arguments have been modified so that $|n|$, $|n'|\le N+1$.
Here $\hat{f}$ denotes a discrete Fourier transform of an analytic function whose explicit form depends on the ``block'' $a$:
for all $x,y\in \Z$, we define
\begin{align} \label{eq:deffanu}
\hat{f}_{a,\nu}(x',x) := \frac{1}{(2\pi)^{2}} \int_{-\pi}^\pi \rmd p' \int_{-\pi}^\pi \rmd p\, f_{a,\nu}(p',p) \rme^{\ci (p'x'+p x)}\,
\end{align}
with $z=1-\cos p$, $w=1-\cos p'$, and
\begin{align}
 & f_{1,\nu}(p',p) := \frac{1}{(z-w)^2+z+w+\nu^2}\, ,\\
 & f_{2,\nu}(p',p) := \frac{z-w}{(z-w)^2+z+w+\nu^2}\, ,\\
 & f_{3,\nu}(p',p) := 1-\frac{(z-w)^2}{(z-w)^2+z+w+\nu^2}\, .
\end{align}
Therefore, for all $x,y\in\{0,1,\ldots,N\}$, $j\in\{1,2,\ldots,N\}$, we also have (recall that $r_x=q_{x+1}-q_x$)
\begin{align}
 & \mean{r_x r_y} = \sum_{n=1}^N B^{(n)}_{5,\nu}(x,y) T_n\,, \quad
 \mean{r_x p_j} = \sum_{n=1}^N B^{(n)}_{6,\nu}(x,j) T_n \, ,
\end{align}
where $B^{(n)}_{5,\nu}$ and $B^{(n)}_{6,\nu}$ are defined as above, using
\begin{align}
 & f_{5,\nu}(p',p) := (\rme^{\ci p'}-1)(\rme^{\ci p}-1)\frac{1}{(z-w)^2+z+w+\nu^2}\, ,\\
 & f_{6,\nu}(p',p) := (\rme^{\ci p'}-1) \frac{z-w}{(z-w)^2+z+w+\nu^2}\, .
\end{align}
(These formulae hold for $x,y\in \{0,N\}$ since any of the previous functions $g$ is $2(N{+}1)$-periodic and symmetric under $x\to -x$, separately in both of its arguments, which implies that $B^{(n)}_{1,\nu}(i,j)$ and $B^{(n)}_{2,\nu}(i,j)$ vanish if either $i$ or $j$ belongs to $\{0,N+1\}$.)

We only need to consider $0<\nu<L^{-1}\le 1/4$, since $\nu\to 0$ is to be taken before $L\to \infty$.  In all of the following, we assume this to hold.
Since for all $|p|\le \pi$ we have $|\sin (p/2)|\ge |p|/\pi$, it follows that $z\ge 2 \pi^{-2} p^2$, and thus
$(z-w)^2+z+w+\nu^2\ge 2 \pi^{-2} (p^2 + (p')^2 + \nu_0^2)$, $\nu_0 = \pi \nu/\sqrt{2}$ on the whole integration region.  Similarly, as $|\sin x|\le |x|$ for all $x$, $|z-w|\le (p^2 + (p')^2)/2$.  We also have for $p=\alpha+\ci \beta$, $\beta\ge 0$, $\alpha\in \R$, that
$|\rme^{\ci p}-1|=|\rme^{-\beta}-\rme^{-\ci \alpha}|\le \beta + |\alpha|$ and thus also
$|\rme^{-\ci p}-1|\le \rme^\beta (\beta + |\alpha|)$.
Such bounds will be used frequently and without further mention in the following.

The worst decay estimate will be obtained when $f=f_{5,\nu}$.  In this case, it is likely that for $\nu=0$ the sum in (\ref{eq:defgnn}) is no longer absolutely convergent, which will complicate the analysis.  Since then $g(n,n')=g(n',n)$ it will be sufficent to consider the case $|n|\le |n'|\le N+1$, which we shall do in the following.  We begin by inserting the identity $1=\cf[|m'|{=}|m|]+\cf[|m'|{<}|m|]+\cf[|m'|{>}|m|]$, which leads to the following decomposition
\begin{align}\label{eq:gnnsplit}
  & g(n',n) =  \hat{f}(n',n) + \sum_{\sigma,\sigma' \in \{\pm 1\} } \sum_{m=1}^\infty \hat{f}(n'+\sigma' m L,n+ \sigma m L)
 \\ \nonumber & \quad + \sum_{\sigma' \in \{\pm 1\} } \sum_{M=0}^\infty \sum_{m=-M}^M
 \Bigl(\hat{f}(n' + \sigma' (M+1) L,n+m L)+  \\  \nonumber & \qquad\quad
\hat{f}(n + \sigma' (M+1) L,n'+m L)\Bigr)
\, .
\end{align}
By construction, in the last sum the first argument has always a larger magnitude than the second:  since $|m|\le M$, we have
$|n'+\sigma'(M+1) L|\ge M L+L-|n'| \ge |m|L+L-|n'| \ge |n+m L|+L-|n'|-|n|\ge |n+m L|$.

Let us then estimate the first term, $\hat{f}(n',n)$, which will turn out to be dominant.  We denote $\sigma'=\text{sign}(n')$, $\sigma=\text{sign}(n)$,
$k'=|n'|$, and $k=|n|$, and recall that then $k\le k'$.  By a change of variables $p\to \sigma p$,
$p'\to \sigma' p'$, which leave $z$ and $w$ invariant,
we thus have
\begin{align*}
& \hat{f}(n',n)= \int_{|p|,|p'|\le \pi} \frac{\rmd p' \rmd p}{(2\pi)^{2}}\,
 \rme^{\ci (k p+ k' p')} \frac{(\rme^{\ci \sigma ' p'}-1)(\rme^{\ci \sigma p}-1)}{(z-w)^2+z+w+\nu^2}
\, .
\end{align*}
Computing the roots of the polynomial in the denominator yields
\begin{align}\label{eq:polfact}
 & (w-z)^2+w+z+\nu^2= (w-w_-)(w-w_+),\quad \text{ where }\\ \nonumber
 & w_\pm = z - \frac{1}{2} \pm \frac{1}{2}\sqrt{1-8 z- 4 \nu^2}\, .
\end{align}
Here $\sqrt{z}$ denotes the principal branch of the complex logarithm, with values such that $\re \sqrt{z}\ge 0$ and $\sqrt{-1}=\ci$.
Then $w_+$ is always at least a distance $\nu^2$ away from $[0,2]$ and $w_-$ always at least a distance $\frac{1}{4}$ away; more precisely
\begin{enumerate}
  \item If $0<z\le \nu^2$, then $ -1\le w_-\le -\frac{3}{8}$ and $-6 \nu^2 \le w_+ \le -\nu^2$.
  \item If $\nu^2<z\le \frac{1}{8}-\frac{\nu^2}{2}$, then $ -1\le w_-\le -\frac{3}{8}$ and $-4 z \le w_+ \le -z$.
  \item If $\frac{1}{8}-\frac{\nu^2}{2}<z\le \frac{1}{4}$, then $\re w_\pm \le -\frac{1}{4}$.
  \item If $z> \frac{1}{4}$, then $|\im w_\pm| > \frac{1}{2}$.
\end{enumerate}

For $z\le \frac{1}{16}<\frac{1}{8}-\frac{\nu^2}{2}$ we define $\bar{z}= \max(z,\nu^2)$, which implies that then $-w_+\ge \bar{z}$.  Then
if $p'=\alpha+\ci \beta$ with $\beta=\sqrt{\bar{z}/2}$, we have $\beta\le 1$, and thus also
$\cosh \beta \le 1+\beta^2 \le 1+\bar{z}/2$.  Therefore, then there are pure constants $C',C$ that $|(w-z)^2+w+z+\nu^2|\ge C' (1-\cos \alpha + \bar{z}/2) \ge C (\alpha^2+\beta^2)$.  By possibly adjusting constants, this inequality continues to hold for all $z$ and $\alpha$, after we choose
$\bar{z}=z_0$ for $z> \frac{1}{16}$, where $z_0>0$ is a suitably small but fixed constant.  Then there is a constant $c_0>0$ such that $\beta \ge c_0 |p|$ for all $|p|\le \pi$.  Therefore,
\begin{align*}
& |\hat{f}(n',n)| \le
C_1 \int_{|p|\le \pi} \rmd p\,  |p| \rme^{-\beta k'} \int_{-\pi}^\pi \rmd \alpha\, \rme^{\beta}\frac{|\alpha|+\beta}{\alpha^2+\beta^2} \\
& \quad \le
C_2 \int_{0}^{\pi} \rmd p\, p \rme^{-\beta k'}
\int_{0}^{\pi/\beta} \rmd s\, \frac{1+s}{1+s^2} \\
& \quad \le
C_3 \int_{0}^{\pi} \rmd p\, p \rme^{-c_0 p k'} (1+\ln p^{-1}) \\
& \quad \le
C_4 (1+k')^{-2} (1+\ln (1+k')) \, .
\end{align*}
This proves that the contribution of the first term has an appropriate upper bound.  Since both of the arguments in the sum in the second term have absolute values between $(m-\frac{1}{2})L$ and $(m+\frac{1}{2})L$, this bound also proves that the second term is bounded by
\begin{align*}
C'_4 L^{-2} (1+\ln L) \le C_4' (1+|n'|)^{-2} (1+\ln (1+|n'|))\, ,
\end{align*}
which is also of the appropriate form.

Let us then consider the first term in the final sum in (\ref{eq:gnnsplit}).  We begin the analysis by summing over $m$, which yields the Dirichlet kernel in the integrand, amounting to changing inside the defining integral
\begin{align*}
  \rme^{\ci n p} \to  \sum_{m=-M}^M \rme^{\ci (n+mL) p} =  \rme^{\ci n p}  D_M(L p)\, , \text{ where }
D_M(x) := \frac{\sin((M+\frac{1}{2}) x)}{\sin(x/2)}\, .
\end{align*}
Since $D_M(x)$ is an even function, this yields the following explicit integral to be considered
\begin{align*}
& \int_{|p|,|p'|\le \pi} \frac{\rmd p' \rmd p}{(2\pi)^{2}}\, D_M(L p) \left[\cos( (n+1) p) - \cos(n p)\right] \\ & \quad \times
 \rme^{\ci (M L + L+\sigma' n') p'} \frac{\rme^{\ci \sigma' p'}-1}{(z-w)^2+z+w+\nu^2}
\, .
\end{align*}
As $|D_M(L p)|\le 2 M +1$ and $|\cos( (n+1) p) - \cos(n p)|\le (1+|n|) p^2$, by using the above deformation of the $p'$ integration contour, we obtain that it is bounded by
\begin{align*}
C_5 M (1+|n|) (1+k')^{-3} (1+\ln k') \, ,
\end{align*}
where $k' = M L + L+\sigma' n'\ge L(M+1/2)$.  As the sum of $\ln (1+M)/(1+M)^2$ over $M=0,1,\ldots$ is finite and $|n|\le L/2$, this proves the the first term in the final sum in (\ref{eq:gnnsplit}) is bounded by
\begin{align*}
C_6 L^{-2} (1+\ln L) \le C_6 (1+|n'|)^{-2} (1+\ln (1+|n'|))\, .
\end{align*}
Since the only property of $n,n'$ used above was that they are bounded by $L/2$, the same bound is valid also for the second term in the final sum.

The function $r\mapsto r^2 (1+\ln r^{-1})$ is increasing for $r\in [0,1]$ and thus the above bound is decreasing in $|n'|$.  Since $\max(|n|,|n'|)\ge (|n|+|n'|)/2$,
we can conclude that there is a constant $c_5$ such that for all $n,n'\in \Z$ we have
\begin{align}
  |g(n',n)|\le c_5  M^{-2} (1+\ln M)\, ,
\end{align}
where $M=1+|n \bmod L|+|n' \bmod L|$.  This is the analogue of the exponential bound quoted in (\ref{eq:fnexpbound}).  Therefore, applying this bound and the inequality in (\ref{eq:persum}) to (\ref{eq:defbna}) proves the first of the bounds stated in (\ref{eq:rxcorrdecay}), namely that there is $c_5'$ such that
\begin{align}
 & |B^{(n)}_{5,0}(i,j)| \le c'_5 m^{-2} (1+\ln m)\, ,
\end{align}
where $m=1+|n-i|+|n-j|+|i-j|$.

For the other two terms ($a=3,6$) the analysis is slightly easier, since the bounds will be summable over the Fourier indices.  Consider the integral defining $\hat{f}_{a,\nu}(n,'n)$, and swap the integration signs as before.  If $a=6$ and $|n'|\ge |n|$, we decompose the integrand first into regions with $z> \frac{1}{16}$ and $z\le \frac{1}{16}$.  For
$z\le \frac{1}{16}$ we decompose further
using (\ref{eq:polfact}) and $z-w=z-w_+ - (w-w_+)$, yielding
\begin{align*}
 & \frac{z-w}{(w-z)^2+w+z+\nu^2} = \frac{z-w_+}{\sqrt{1-8z-4\nu^2}} \frac{1}{w-w_+} \\
&\quad -\frac{1}{w-w_-}\left(\frac{z-w_+}{\sqrt{1-8z-4\nu^2}}+1\right)\, .
\end{align*}
Here for the second piece, as well as for the case $z> \frac{1}{16}$, we can then shift the integration contour of $p'$ using $\beta=\beta_0$ which is a small fixed constant.  These terms have an exponential bound $C\rme^{-\beta_0 |n'|}$.  The remaining piece we estimate using the same $\beta=\beta(z)$ as before, which proves that
\begin{align*}
& |\hat{f}_{6,\nu}(n',n)| \le
C \left(\rme^{-\beta_0 |n'|} + \int_{|p|\le \frac{1}{16}} \rmd p\,  (z+\nu^2) \rme^{-\beta |n'|} \int_{-\pi}^\pi \rmd \alpha\, \rme^{\beta}\frac{|\alpha|+\beta}{\alpha^2+\beta^2}
\right) \\
& \quad \le
C \left(\rme^{-\beta_0 |n'|} + \nu^3 \ln \nu^{-1} \rme^{-\nu |n'|} + (1+|n'|)^{-3} (1+\ln (1+|n'|))
\right) \, ,
\end{align*}
where we have used the facts that for $|p|\le \frac{1}{16}$ we have $|z-w_+|\le C (z+\nu^2)$, and that $\beta=\nu$ and $z\le \nu^2$ if $|p|\le \nu$.
If $a=6$ and $|n|> |n'|$, we need to swap the order of integration, but we can still apply the same decompositions as above yielding
\begin{align*}
& |\hat{f}_{6,\nu}(n',n)| \le
C \left(\rme^{-\beta_0 |n|} + \int_{|p|\le \frac{1}{16}} \rmd p\,  (z+\nu^2) |p| \rme^{-\beta |n|} \int_{-\pi}^\pi \rmd \alpha\, \frac{1}{\alpha^2+\beta^2}
\right) \\
& \quad \le
C \left(\rme^{-\beta_0 |n|} + \nu^3 \rme^{-\nu |n|} + (1+|n|)^{-3}
\right) \, .
\end{align*}
Therefore, $|\hat{f}_{6,\nu}(n',n)| \le C [\rme^{-\beta_0 k'} + \nu^3 \ln \nu^{-1} \rme^{-\nu k'} + (1+k')^{-3} (1+\ln (1+k'))]$ where $k'=\max(|n|,|n'|)$, and thus
\begin{align*}
& |g(n',n)| \le  C [M^{-3} (1+\ln M)+\rme^{-\beta_0 M/2} + L^{-2} \nu \ln \nu^{-1} + L^{-3} (1+\ln L)]
\end{align*}
where $M=1+|n \bmod L|+|n' \bmod L|$, as before.  Taking here the limit $\nu\to 0$ and then following the same steps as before proves that a constant for the second of the bounds ($a=6$) stated in (\ref{eq:rxcorrdecay}) can be found.

For the final remaining case $a=3$ we can again use symmetry and assume from the beginning that $|n'|\ge |n|$.  We decompose then for $z\le \frac{1}{16}$
\begin{align}
 & f_{3,\nu}(p',p) = -\frac{(z-w_+)^2}{\sqrt{1-8z-4\nu^2}} \frac{1}{w-w_+} \\ \nonumber
&\quad +\frac{1}{w-w_-}\left(\frac{(z-w_+)^2}{\sqrt{1-8z-4\nu^2}}+2 (z-w_+)+w_+ -w_-\right)
\end{align}
and apply the rest of the argument as in the previous case.  This yields a bound
$|\hat{f}_{3,\nu}(n',n)| \le C [\rme^{-\beta_0 k'} + \nu^4 \rme^{-\nu k'} + (1+k')^{-4}]$ where $k'=\max(|n|,|n'|)$, and thus
\begin{align*}
& |g(n',n)| \le  C [M^{-4} + \rme^{-\beta_0 M/2} + L^{-2} \nu^2 + L^{-4}]
\end{align*}
where $M=1+|n \bmod L|+|n' \bmod L|$.  As above, this implies the final bound in (\ref{eq:rxcorrdecay}), for $a=3$.

\end{document}